%% file: Part1Smooth_v9.tex
\documentclass[letterpaper, 11pt]{llncs}

\usepackage[margin=1in]{geometry}
\usepackage[bookmarks, colorlinks=true, plainpages = false, citecolor = blue,linkcolor=red,urlcolor = blue, filecolor = blue]{hyperref}
\usepackage{url}

\usepackage{amsmath,amsfonts,amssymb,bm, verbatim,mathtools}

\usepackage{color,graphicx,appendix}

\usepackage{subfigure}
\usepackage{etoolbox}
\usepackage{tikz}
\usepackage{xr,xspace}
\usepackage{todonotes}
\usepackage{paralist}
\usepackage{caption}

\newcommand{\mycomment}[1]{}
\newcommand{\eproof}{\hfill $\Box$}

\usepackage{xspace,prettyref}
\newcommand{\diverge}{\to\infty}

\newcommand{\reals}{{\mathbb{R}}}

\newrefformat{eq}{(\ref{#1})}
\newrefformat{chap}{Chapter~\ref{#1}}
\newrefformat{sec}{Section~\ref{#1}}
\newrefformat{algo}{Algorithm~\ref{#1}}
\newrefformat{fig}{Fig.~\ref{#1}}
\newrefformat{tab}{Table~\ref{#1}}
\newrefformat{rmk}{Remark~\ref{#1}}
\newrefformat{clm}{Claim~\ref{#1}}
\newrefformat{def}{Definition~\ref{#1}}
\newrefformat{cor}{Corollary~\ref{#1}}
\newrefformat{lmm}{Lemma~\ref{#1}}
\newrefformat{prop}{Proposition~\ref{#1}}
\newrefformat{app}{Appendix~\ref{#1}}
\newrefformat{hyp}{Hypothesis~\ref{#1}}
\newrefformat{thm}{Theorem~\ref{#1}}

\newcommand{\pth}[1]{\left( #1 \right)}

\newcommand{\norm}[1]{\left\|{#1} \right\|}

\newcommand{\tx}{{\widetilde{x}}}

\newcommand{\calA}{{\mathcal{A}}}

\newcommand{\calC}{{\mathcal{C}}}

\newcommand{\calF}{{\mathcal{F}}}

\newcommand{\calI}{{\mathcal{I}}}

\newcommand{\calL}{{\mathcal{L}}}

\newcommand{\calN}{{\mathcal{N}}}

\newcommand{\calR}{{\mathcal{R}}}
\newcommand{\calS}{{\mathcal{S}}}

\newcommand{\calV}{{\mathcal{V}}}

\newcommand{\argmin}{{\rm argmin}}

\begin{document}

\title{Byzantine Multi-Agent Optimization -- Part I\,\thanks{This research is supported in part by National Science Foundation awards NSF 1329681 and 1421918.
Any opinions, findings, and conclusions or recommendations expressed here are those of the authors
and do not necessarily reflect the views of the funding agencies or the U.S. government.}}

\author{Lili Su \hspace*{1in} Nitin Vaidya}
\institute{Department of Electrical and Computer Engineering, and\\
Coordinated Science Laboratory\\
University of Illinois at Urbana-Champaign\\
Email:\{lilisu3, nhv\}@illinois.edu}
\maketitle

\begin{center}
Technical Report\\
~\\

June 2015\footnote{Revised in July 2015}
\end{center}
~

\centerline{\bf Abstract}

~
We study Byzantine fault-tolerant distributed optimization of a sum of convex (cost) functions with real-valued scalar
input/ouput. In particular, the goal is to optimize
a global cost function $\frac{1}{|\calN|}\sum_{i\in \calN} h_i(x)$, where $\calN$ is the set of non-faulty
agents, and $h_i(x)$ is agent $i$'s local cost function, which is initially
known only to agent $i$. In general,
when some of the agents may be Byzantine faulty, the above goal is unachievable,
because the identity of the faulty agents is not necessarily known to the non-faulty agents, and the faulty
agents may behave arbitrarily.
Since the above global cost function cannot be optimized exactly in presence of Byzantine agents,
we define a weaker version of the problem.\\

The goal for the weaker problem is to generate an output that is an optimum of a function formed as
a {\em convex combination} of local cost functions of the non-faulty agents. More precisely, for some choice
of weights $\alpha_i$ for $i\in \calN$ such that $\alpha_i\geq 0$ and $\sum_{i\in \calN}\alpha_i=1$,
the output must be an optimum of the cost function
$\sum_{i\in \calN} \alpha_ih_i(x)$. Ideally, we would like $\alpha_i=\frac{1}{|\calN|}$ for all $i\in \calN$ -- however,
this cannot be guaranteed due to the presence of faulty agents.
In fact, we show that the maximum achievable number of nonzero weights
($\alpha_i$'s) is $|\calN|-f$, where $f$ is the upper bound on the number of Byzantine agents.
In addition, we present algorithms that ensure that at least $|\calN|-f$ agents have weights that are
bounded away from 0. A low-complexity suboptimal algorithm is proposed, which ensures that at least $\lceil \frac{n}{2}\rceil-\phi$ agents have weights that are
bounded away from 0, where $n$ is the total number of agents, and $\phi$ ($\phi\le f$) is the actual number of Byzantine agents.

%

~

~

%
%
%

\newpage

\section{System Model and Problem Formulation}
\label{sec:intro}

The system under consideration is synchronous, and consists of $n$ agents connected by a complete communication network.
The set of agents is $\calV=\{1,\cdots,n\}$.
We assume that $n>3f$ for reasons that will be clearer soon.
%
%
We say that a function $h: \mathbb{R}\rightarrow \mathbb{R}$ is {\em admissible} if (i) $h(\cdot)$ is convex, and continuously
differentiable,
and (ii) the set $\arg\min_{x\in\mathbb{R}} h(x)$ containing the optima of $h(\cdot)$
is non-empty and compact (i.e., bounded and closed).
Each agent $i\in \calV$ is initially provided with an {\em admissible} local cost function $h_i: \mathbb{R}\rightarrow\mathbb{R}$.

Up to $f$ of the $n$ agents may be Byzantine faulty.
Let $\calF$ denote the set of faulty agents, and let $\calN = \calV - \calF$ denote the
set of non-faulty agents. The set $\calF$ of faulty agents may be chosen by an adversary arbitrarily. Let $|\calF|=\phi$. Note that $\phi\leq f$ and $|\calN|\geq n-f$.\\

The ideal goal here is to develop algorithms that optimize the {\em average} of
the local cost functions at the non-faulty agents, and allow the non-faulty
agents to reach consensus on an optimum $x$.
Thus, ideally, as stated in Problem 1 in Figure \ref{fig:prob},
each non-faulty agent should output an {\em identical} value $\tx\in\mathbb{R}$
that minimizes $\frac{1}{|\calN|}\sum_{i\in \calN} h_i(x)$.

\begin{figure}
\begin{tabular}{|l|l|l|}\hline
\begin{minipage}[t]{0.3\textwidth}
{\bf Problem 1}\\
\begin{align}
\nonumber
\tx\in\arg \min_{x\in\mathbb{R}}\quad \frac{1}{|\calN|}\sum_{i\in \calN} h_i(x)
\end{align}
\end{minipage}
&
\begin{minipage}[t]{0.30\textwidth}
{\bf Problem 2}\\
\begin{eqnarray*}
\tx & \in & \arg \min_{x\in\mathbb{R}}\quad \sum_{i\in \calN} \alpha_i h_i(x)\\
\text{such that} && \nonumber \\
&&\forall i\in\calN, ~ \alpha_i\geq 0 \text{~~and~~} \nonumber \\
&&\sum_{i\in \calN}\alpha_i=1 \nonumber
\end{eqnarray*}
\end{minipage}
&
\begin{minipage}[t]{0.4\textwidth}
{\bf Problem 3 with parameters $\beta,\gamma$, $\beta\geq 0$}\\
\begin{eqnarray*}
\tx & \in & \arg \min_{x\in\mathbb{R}}\quad \sum_{i\in \calN} \alpha_i h_i(x)\\
\text{such that} && \nonumber \\
&&\forall i\in\calN, ~ \alpha_i\geq 0,  \nonumber \\
&&\sum_{i\in \calN}\alpha_i=1, \text{~~and~~} \nonumber \\
&&\sum_{i\in\calN} {\bf 1}(\alpha_i>\beta) ~ \geq ~ \gamma \nonumber
\end{eqnarray*}

\end{minipage}
~\\
\hline
\end{tabular}
\caption{Problem formulations: All non-faulty agents
must output an identical value $\tx\in\mathbb{R}$ that satisfies the constraints specified in each problem formulation.}
\label{fig:prob}
\end{figure}

The presence of Byzantine faulty nodes makes it impossible to design an algorithm that can solve Problem 1 for all admissible local cost functions (this is shown formally in Appendix \ref{appendix:imposs:0}).  Therefore, we introduce a weaker version of the problem, namely, Problem 2 in Figure \ref{fig:prob}.
Problem 2 requires that the output $\tx$ be an optimum of a function formed as
a {\em convex combination} of local cost functions of the non-faulty agents. More precisely, for some choice
of weights $\alpha_i$ for $i\in \calN$ such that $\alpha_i\geq 0$ and $\sum_{i\in \calN}\alpha_i=1$,
the output must be an optimum of the weighted cost function $\sum_{i\in \calN} \alpha_i\,h_i(x)$.
%
%
Ideally, we would like $\alpha_i=\frac{1}{|\calN|}$ for all $i\in \calN$,
since that would effectively solve Problem 1. However, as noted above,
this cannot be guaranteed due to the presence of faulty agents.
Therefore, in general, all $\alpha_i$'s may not necessarily be non-zero for the
chosen solution of Problem 2.
The desired goal then is to {\em maximize} the number of weights ($\alpha_i$'s)
that are bounded away from zero.
With this in mind, we introduce the third, and final problem formulation (Problem 3)
in Figure \ref{fig:prob}.
In Problem 3,
note that
${\bf 1}\{\alpha_i>\beta\}$ is an indicator function that outputs 1 if $\alpha_i>\beta$,
and 0 otherwise.
Essentially, Problem 3 adds a constraint to Problem 2, requiring that at least $\gamma$
weights must exceed a threshold $\beta$, where $\beta\geq 0$. Thus, $\beta,\gamma$ are parameters
of Problem 3.\\

We will say that Problem 1, 2 or 3 is solvable if there exists an algorithm that will
find a solution for the problem (satisfying all its constraints) for all admissible local cost functions,
and all possible behaviors of faulty nodes.
Our problem formulations require that all non-faulty agents output identical $\tx\in\mathbb{R}$, while
satisfying the constraints imposed by the problem (as listed in Figure \ref{fig:prob}).
Thus, the traditional Byzantine consensus \cite{impossible_proof_lynch} problem, which also imposes a similar
{\em agreement} condition, is a special case
of our optimization problem.\footnote{This can be proved formally as follows.
Suppose that we want to solve the Byzantine consensus problem where the input
of agent $i$ is $a_i\in \{0,1\}$. Then defining $f_i(x)=(x-a_i)^2$ ensures that the
output $\tx$ of correct algorithms for Problems 1, 2, 3 will be in the convex hull of
the inputs at the non-faulty agents. Choose $[\,\tx\,]$ as the output for the consensus problem.}
Therefore,
the lower bound of $n>3f$ for Byzantine consensus \cite{impossible_proof_lynch} also applies
to our problem. Hence we assume that $n>3f$.\\

 We prove the following key results:
\begin{itemize}
\item (Theorem \ref{t_imposs_0}) Problem 1 is not solvable when $f>0$.
\item (Theorem \ref{ub1}) For any $\beta\geq 0$, Problem 3 is not solvable if $\gamma>|\calN|-f$.
\item (Theorems \ref{t_algo1}, \ref{t_algo2}, and \ref{converge of alg5}) Problem 3 is solvable with
$\beta\leq \frac{1}{2(|\calN|-f)}$ and $\gamma\leq |\calN|-f$.
\end{itemize}
\vskip 1.5\baselineskip
Our results in Theorems \ref{t_algo1} and \ref{t_algo2} can be strengthened to $\beta\le \frac{1}{n}$, as discussed later.\\

In our other work we explore subdifferentiable functions,  restricted functions families, and asynchronous systems, respectively. Results will be presented in other reports.
\vskip \baselineskip

The rest of the report is organized as follows. Related work is summarized in Section \ref{related work}. Impossibility results, in particular, Theorem \ref{t_imposs_0} and Theorem \ref{ub1} are presented in Section \ref{impossibility}. Achievability of $\gamma=|\calN|-f$ is proved constructively in Section \ref{algorithm smooth} wherein five algorithms are proposed. In particular, Algorithms 1, 2, 3 and 5 solve Problem 3 with $\beta=\frac{1}{2\pth{|\calN|-f}}$ and $\gamma=|\calN|-f$, and Algorithm 4 solves Problem 3 with  $\beta=\frac{1}{2\pth{n-f}}$ and $\gamma=n-2f$. The performance of Algorithm 4 is independent of $|\calN|$ and is, in general, slightly weaker than the performance of Algorithms 1, 2, 3 and 5. Alternative performance analysis of Algorithms 1, 2 and 3 is also presented -- Algorithms 1, 2 and 3 are shown to solve Problem 3 with $\beta=\frac{1}{n}$ and $\gamma=|\calN|-f$.
Section \ref{approximation alg} presents a low-complexity suboptimal algorithm that solves Problem 3 with $\beta=\frac{1}{2|\calN|}$ and $\gamma=\lceil\frac{n}{2}\rceil-\phi$. 
Our other technical reports under preparation that extend the results presented in this report are briefly discussed in Section \ref{extension}. Section \ref{conclusion and discussion} concludes the report.

\section{Related Work}\label{related work}

Fault-tolerant consensus \cite{PeaseShostakLamport} is a special case of the optimization problem considered in this report. There is a significant body of work on fault-tolerant consensus, including \cite{Dolev:1986:RAA:5925.5931,Chaudhuri92morechoices,mostefaoui2003conditions,fekete1990asymptotically,LeBlanc2012,vaidya2012iterative,friedman2007asynchronous}.
The optimization algorithms presented in this report use Byzantine consensus as a component.

Convex optimization, including distributed convex optimization, also has a long history \cite{bertsekas1989parallel}. However, we are not aware of
prior work that obtains the results presented in this report.
Primal and dual decomposition methods that lend themselves naturally to a distributed paradigm are well-known \cite{Boyd2011}. 
There has been significant research on a variant of distributed optimization problem \cite{Duchi2012,Nedic2009,Tsianos2012}, in which the global objective $h(x)$ is a summation of $n$ convex functions, i.e, $h(x)=\sum_{j=1}^n h_j(x)$, with function $h_j(x)$ being known to the $j$-th agent. The need for robustness for distributed optimization problems has received some attentions recently \cite{Duchi2012,kailkhura2015consensus,zhang2014distributed,marano2009distributed}. In particular, Duchi et al.\ \cite{Duchi2012} studied the impact of random communication link failures on the convergence of distributed variant of dual averaging algorithm. Specifically, each realizable link failure pattern considered in \cite{Duchi2012} is assumed to admit a doubly-stochastic matrix which governs
the evolution dynamics of local estimates of the optimum.

In other related work, significant attempts have been made to solve the problem of distributed hypothesis testing in the presence of Byzantine attacks \cite{kailkhura2015consensus,zhang2014distributed,marano2009distributed}, where Byzantine sensors may transmit fictitious observations aimed at confusing the decision maker to arrive at a judgment that is in contrast with the true underlying distribution. Consensus based variant of distributed event detection, where a centralized data fusion center does not exist, is considered in \cite{kailkhura2015consensus}. In contrast, in this paper, we focus on the Byzantine attacks on the multi-agent optimization problem.

\section{Impossibility Results}\label{impossibility}

Recall that we say that Problem $i~(i=1, 2, 3)$ is solvable if there exists an algorithm that will
find a solution for the problem (satisfying all its constraints) for all admissible local cost functions,
and all possible behaviors of faulty nodes.
The intuitive result below, proved in
Appendix \ref{appendix:imposs:0},
shows that there is no solution for Problem 1
in presence of faulty agents.
\begin{theorem}
\label{t_imposs_0}
Problem 1 is not solvable when $f>0$.
\end{theorem}
Theorem \ref{ub1} below presents an upper bound on parameter
$\gamma$ for Problem 3 to be solvable.
\begin{theorem}
\label{ub1}
For any $\beta\geq 0$, Problem 3 is not solvable if
$\gamma>|\calN|-f$.
\end{theorem}
Appendix \ref{appendix:ub1} presents the proof.
In the next section, we will show that the upper bound of Theorem \ref{ub1} is achievable.

%

\section{Proposed Algorithms}\label{algorithm smooth}

In this section, we present five different algorithms.
The first two algorithms, named Algorithm 1 and Algorithm 2, respectively, are not necessarily practical, but allow
us to derive results that are useful in proving the correctness of Algorithm 3, Algorithm 4 and Algorithm 5,
which are more practical. As an alternative to Algorithm 1, Algorithm 2 admits more concise correctness proof than that of Algorithm 1.
Algorithm 4 and Algorithm 5 require less memory than Algorithm 3. However, in contrast to Algorithms 1, 2 and 3,
the local estimates at non-faulty agents converge in neither Algorithm 4 nor Algorithm 5.

\subsection{Algorithm 1}

Algorithm 1 pseudo-code for agent $j$ is presented below.

\paragraph{}
\vspace*{8pt}\hrule
~

{\bf Algorithm 1} for agent $j$:
\vspace*{4pt}\hrule

\begin{list}{}{}
\item[{\bf Step 1:}]
Perform Byzantine broadcast of local cost function\footnote{In this step, each agent $j$ broadcasts a complete description of its cost function to other agents.}
 $h_j(x)$ to all the agents
using any Byzantine broadcast algorithm, such as \cite{psl_BG_1982}.

In step 1,
agent $j$ should receive from each agent $i\in\calV$ its cost function $h_i(x)$.
For non-faulty agent $i\in\calN$, $h_i(x)$ will be an admissible function
({\em admissible} is defined in Section \ref{sec:intro}).
If a faulty agent $k\in\calF$ does not correctly perform Byzantine broadcast of
its cost function, or broadcasts
an inadmissible cost function, then hereafter assume $h_k(x)$
to be a {\em default} admissible cost function that is known to all agents.\\

\item[{\bf Step 2:}]
The multiset of admissible functions obtained in Step 1 is
$\{h_1(x), h_2(x), \cdots, h_n(x)\}$.
For each $x\in \reals$, define multisets $A(x), B(x), C(x)$ below, where $h_i^{\prime}(x)$ denotes the gradient of function $h_i(\cdot)$ at $x$.
\begin{align*}
A\pth{x}\triangleq \{i:~h_i^{\prime}\pth{x}>0\}, \\
B\pth{x}\triangleq \{i:~h_i^{\prime}\pth{x}<0\}, \\
C\pth{x}\triangleq \{i: ~h_i^{\prime}\pth{x}=0\}.
\end{align*}
If there exists $x\in\mathbb{R}$ such that
\begin{align}
\label{alg1.2}
\min_{F_1:~F_1\subseteq A\pth{x}~\text{and}~ |F_1|\le f}\sum_{i\in A\pth{x}-F_1}h_i^{\prime}\pth{x}+\max_{F_2:~F_2\subseteq B\pth{x}~\text{and}~ |F_2|\le f} \sum_{i\in B\pth{x}-F_2}h_i^{\prime}\pth{x}~=~0~~~~
\end{align}
then deterministically choose output $\tx$ to be any one $x$ value that satisfies (\ref{alg1.2});\\otherwise, choose output $\tx=\perp$.
\end{list}

~
\hrule

~
\\

We will prove that Algorithm 1 solves Problem 3 with parameters $\beta=\frac{1}{2(|\calN|-f)}$ and $\gamma=|\calN|-f$.\\

For the multiset $\{h_1(x),h_2(x),\cdots,h_n(x)\}$\, of $n$ admissible cost functions
gathered in Step 1 of Algorithm 1,
 define $F\pth{x}$ and $G\pth{x}$ as follows,
where $A(x)$, $B(x)$ and $C(x)$ are as defined in Algorithm 1.

$$F\pth{x}\triangleq \min_{F_1:~F_1\subseteq A\pth{x}~\text{and}~ |F_1|\le f}\sum_{i\in A\pth{x}-F_1}h_i^{\prime}\pth{x}$$ and
$$G\pth{x}\triangleq \max_{F_2:~F_2\subseteq B\pth{x}~\text{and}~ |F_2|\le f}\sum_{i\in B\pth{x}-F_2}h_i^{\prime}\pth{x}.$$

\begin{proposition}
\label{monotone}
 $F\pth{x}$ and $G\pth{x}$ are both non-decreasing functions of $x\in\mathbb{R}$.
\end{proposition}
The proof of Proposition \ref{monotone} is presented in Appendix \ref{appendix:monotone}.

\begin{proposition}
\label{continuous}
Both $F\pth{x}$ and $G\pth{x}$ are continuous functions of $x\in\mathbb{R}$.
\end{proposition}
The proof of Proposition \ref{continuous} is presented in Appendix \ref{appendix:continuous}.

\begin{lemma}
\label{well-definedOutPut}
Algorithm 1 returns $\tx\in\mathbb{R}$ when $n>3f$ (i.e., it does not return $\bot$).
\end{lemma}
\begin{proof}
If there exists $x\in\mathbb{R}$ that satisfies equation (\ref{alg1.2}) in
Algorithm 1, then the algorithm will not return $\bot$. Thus,
to prove this lemma, it suffices to show that there exists $x\in\mathbb{R}$ that satisfies equation (\ref{alg1.2}).
Consider the multiset of admissible functions $\{h_1(x),h_2(x),\cdots,h_n(x)\}$ obtained by a non-faulty
agent in Step 1 of Algorithm 1.
Define $X_i =\arg\min_{x\in\mathbb{R}} h_i(x)$.
Let $\max X_i$ and $\min X_i$ denote the largest and smallest values in $X_i$, respectively.
Sort the above $n$ functions $h_i(x)$ in an {\em increasing} order of their $\max X_i$ values, breaking
ties arbitrarily. Let $i_0$ denote the $f+1$-th agent in this sorted order (i.e., $i_0$ has
the $f+1$-th smallest value in the above sorted order).
Similarly,
sort the functions $h_i(x)$ in an {\em decreasing} order of $\min X_i$ values, breaking
ties arbitrarily. Let $j_0$ denote the $f+1$-th agent in this sorted order (i.e., $j_0$ has
the $f+1$-th largest value in the above sorted order).
Define function $H(\cdot)$ as
$$H\pth{x}=F\pth{x}+G\pth{x}.$$
Consider $x_1\in X_{i_0}$ and $x_2\in X_{j_0}$.
Then, by the definition of $i_0$, $j_0$, $F(\cdot)$ and $G(\cdot)$, we have
\begin{eqnarray*}
&&H\pth{x_1}~=F\pth{x_1}+G\pth{x_1} ~=~0+G\pth{x_1} ~\le~0, \\
\text{~~~~~and\hspace*{0.6in}} &&\\
&&H\pth{x_2}~=~F\pth{x_2}+G\pth{x_2}~ =~F\pth{x_2}+0~ \ge~ 0
\end{eqnarray*}
If $H(x_1)=0$ or $H(x_2)=0$, then $x_1$ or $x_2$, respectively,
satisfy equation (\ref{alg1.2}), proving the lemma.
(Note that $H(\cdot)=F(\cdot)+G(\cdot)$, and the definition of
$F(\cdot)$ and $G(\cdot)$ implies that, if $H(x_i)=0$ then
$x_i$ satisfies equation (\ref{alg1.2}).

Let us now consider the case when $H(x_1)<0$ and $H(x_2)>0$.
By Propositions \ref{monotone} and \ref{continuous},
we have that $H\pth{\cdot}$ is non-decreasing and continuous. Then it follows that $x_1\le x_2$, and there exists $\tx\in [x_1, x_2]$ such that
$H\pth{\tx}=0$, i.e., $\tx$ satisfies equation (\ref{alg1.2}), proving the lemma.

\eproof
\end{proof}

~

The next two theorems prove that Algorithm 1 can solve Problem 3 for
$\gamma=|\calN|-f$, proving that the bound on $\gamma$ stated in Theorem \ref{ub1} is tight
for certain values of $\beta$ (as stated in the theorem below).


\begin{theorem}
\label{t_algo1}
When $n>3f$,
Algorithm $1$ solves Problem 3 with $\beta=\frac{1}{2(|\calN|-f)}$ and $\gamma=|\calN|-f$.
\end{theorem}
\begin{proof}
By Lemma \ref{well-definedOutPut}, we know that Algorithm $1$ returns a value in $\mathbb{R}$.
Let $\tx$ be the output of Algorithm $1$ for the set of functions $\{h_1(x),h_2(x),\cdots,h_n(x)\}$
gathered in Step 1 of the algorithm.
Consider $F_1^*\subseteq A(\tx)$ and $F_2^*\subseteq{B(\tx)}$,  with $|F_1^*|\le f$ and $|F_2^*|\le f$, that
minimize $\sum_{i\in A\pth{\tx}-F_1} h_i^{\prime}\pth{\tx}$, and
maximize $\sum_{i\in B\pth{\tx}-F_2} h_i^{\prime}\pth{\tx}$, respectively
(as per equation (\ref{alg1.2})).

Recall that $\calV= \{1, \ldots, n\}$. Sort the elements in the multiset $\{h_1^{\prime}(\tx), \ldots, h_n^{\prime}(\tx)\}$ in a non-increasing order, breaking ties in such a way that the elements corresponding to the agents in $F_1^*$ are among the first $f$ elements in the sorted order and the elements corresponding to the agents in $F_2^*$ are among the last $f$ elements in the sorted order. Such a sorted order is well-defined since $|F_1^*|\le f$ and $|F_2^*|\le f$.
Let $\bar{F_1}\subseteq \calV$ be the agents corresponding to the first $f$ elements in the sorted order, and let $\bar{F_2}\subseteq \calV$ be the agents corresponding to the last $f$ elements in the sorted order. Note that $F_1^*\subseteq \bar{F_1}$ and $F_2^*\subseteq \bar{F_2}$.
Since $A\pth{\tx}, B\pth{\tx}$ and $C\pth{\tx}$ form a partition of $\calV$, we have
\begin{align}
\nonumber
\sum_{i\in \calV-F_1^*-F_2^*} h_i^{\prime}\pth{\tx}&=\sum_{i\in C\pth{\tx}}h_i^{\prime}\pth{\tx}+\sum_{i\in A\pth{\tx}\cup B\pth{\tx}-F_1^*-F_2^*}h_i^{\prime}\pth{\tx}\\
\nonumber
&\overset{(a)}{=}0+\sum_{i\in A\pth{\tx}\cup B\pth{\tx}-F_1^*-F_2^*}h_i^{\prime}\pth{\tx}\\
&\overset{(b)}{=}0+0=0
\label{R0}
\end{align}
Equality $(a)$ follows by definition of $C\pth{\tx}$, and equality $(b)$ is true
because $\tx$ satisfies equation (\ref{alg1.2}).
Denote $\calR^*=\calV-\bar{F_1}-\bar{F_2}$. Next we show that
\begin{align}
\label{R bar}
\sum_{i\in \calR^*} h_i^{\prime}\pth{\tx}=0.
\end{align}
If $|A(\tx)|\ge f$, by definition of $F_1^*$, it holds that $|F_1^*|=f$. Thus, $\bar{F_1}=F_1^*$. Consequently, we have
$$\sum_{i\in \bar{F_1}-F_1^*}h_i^{\prime}(\tx)=\sum_{i\in \O}h_i^{\prime}(\tx)=0.$$
If $|A(\tx)|< f$, by definition of $F_1^*$ and $\bar{F_1}$, and the fact that $F_1^*\subset \bar{F_1}$, it follows that $F_1^*=A(\tx)$, and $h_i^{\prime}(\tx)\le 0$ for each $i\in \bar{F_1}-F_1^*=\bar{F_1}-A(\tx)\not=\O$. In addition, if there exists $i\in \bar{F_1}-F_1^*$ such that $h_i^{\prime}(\tx)< 0$, then by definition of $\bar{F_1}$, we have $h_j^{\prime}(\tx)< 0$ for each $j\in \calV-\bar{F_1}$. So we get
\begin{align*}
0&=\sum_{i\in \calV-F_1^*-F_2^*} h_i^{\prime}\pth{\tx}~~~~\text{by}~(\ref{R0})\\
&=\sum_{i\in \calV-\bar{F}_1-F_2^*} h_i^{\prime}\pth{\tx}+\sum_{i\in \bar{F_1}-F_1^*} h_i^{\prime}\pth{\tx}~~~\text{since}~F_1^*\subseteq \bar{F_1}\\
&\le \sum_{i\in \calV-\bar{F}_1-F_2^*} h_i^{\prime}\pth{\tx}~~\text{since}~h_i^{\prime}(\tx)\le 0~\text{for each}~i\in \bar{F_1}-F_1^*\\
&<0~~\text{since}~h_i^{\prime}(\tx)< 0~\text{for each}~i\in \calV-\bar{F}_1,
\end{align*}
proving a contradiction. Thus, there does not exist $i\in \bar{F_1}-F_1^*$ such that $h_i^{\prime}(\tx)< 0$, i.e.,  $h_i^{\prime}(\tx)= 0$ for each $i\in \bar{F_1}-F_1^*$. Consequently, we have
$$\sum_{i\in \bar{F_1}-F_1^*}h_i^{\prime}(\tx)=\sum_{i\in \bar{F_1}-F_1^*} 0=0.$$
Hence, regardless of the size of $|A(\tx)|$, the following is always true.
\begin{align}
\label{pp1}
\sum_{i\in \bar{F_1}-F_1^*}h_i^{\prime}(\tx)=0.
\end{align}
Similarly, we can show that
\begin{align}
\label{pp2}
\sum_{i\in \bar{F_2}-F_2^*}h_i^{\prime}(\tx)=0.
\end{align}
Therefore, we have
\begin{align*}
0&=\sum_{i\in \calV-F_1^*-F_2^*} h_i^{\prime}\pth{\tx}~~~~\text{by}~(\ref{R0})\\
&=\sum_{i\in \calV-\bar{F_1}-\bar{F_2}} h_i^{\prime}\pth{\tx}+\sum_{i\in \bar{F_1}-F_1^*} h_i^{\prime}\pth{\tx}+
\sum_{i\in \bar{F_2}-F_2^*} h_i^{\prime}\pth{\tx}~~~\text{since}~F_1^*\subseteq \bar{F_1}~\text{and}~F_2^*\subseteq \bar{F_2}\\
&=\sum_{i\in \calR^*} h_i^{\prime}\pth{\tx}+\sum_{i\in \bar{F_1}-F_1^*} h_i^{\prime}\pth{\tx}+
\sum_{i\in \bar{F_2}-F_2^*} h_i^{\prime}\pth{\tx}~~~~\text{by definition of}~\calR^*\\
&=\sum_{i\in \calR^*} h_i^{\prime}\pth{\tx}+0+0~~~~\text{by}~(\ref{pp1})~\text{and}~(\ref{pp2})\\
&=\sum_{i\in \calR^*} h_i^{\prime}\pth{\tx},
\end{align*}
proving equation (\ref{R bar}).\\

Let $\widetilde{F}_1\subseteq \bar{F_1}-\calF$ and $\widetilde{F}_2\subseteq \bar{F_2}-\calF$ such that
\begin{align}
\label{reintroduction 1}
|\widetilde{F}_1|= f-\phi+|\calR^*\cap \calF|~ ~\text{and} ~~ |\widetilde{F}_2|= f-\phi+|\calR^*\cap \calF|.
\end{align}
Since $|\calF|=\phi\le f$, $|\bar{F_1}|=f=|\bar{F_2}|$, and $\calR^*\cup \bar{F_1}\cup \bar{F_2}=\calV$, it holds that
$$|\bar{F_1}-\calF|\ge f-\phi+|\calR^*\cap \calF| ~~~\text{and}~~~|\bar{F_2}-\calF|\ge f-\phi+|\calR^*\cap \calF|.$$ Thus, $\widetilde{F}_1$ and $\widetilde{F}_2$ are well-defined.

We now show that
\begin{align}
\label{boundary sign}
\sum_{i\in \widetilde{F}_1}h_i^{\prime}\pth{\tx}\ge 0~~\text{and}~~\sum_{i\in \widetilde{F}_2}h_i^{\prime}\pth{\tx}\le 0.
\end{align}

Suppose $\sum_{i\in \widetilde{F}_1}h_i^{\prime}\pth{\tx}< 0$, then there exists $i_0\in \widetilde{F}_1\subseteq \bar{F_1}-\calF$ such that $h_{i_0}^{\prime}\pth{\tx}<0$. Since agents in $\bar{F_1}$ have the $f$ largest values (including ties) in the set $\{h_1^{\prime}(\tx), \ldots, h_n^{\prime}(\tx)\}$, then $h_i^{\prime}(\tx)<0$ for each $i\in \calR^*$, contradicting the fact that (\ref{R bar}) holds.   Analogously, it can be shown that $\sum_{i\in \widetilde{F}_2}h_i^{\prime}\pth{\tx}\le 0$.

In addition, we observe that
\begin{align}
\label{sandwich}
\sum_{i\in \widetilde{F}_2}h_i^{\prime}\pth{\tx}\le \sum_{i\in \calR^*\cap\calF} h_i^{\prime}\pth{\tx}\le \sum_{i\in \widetilde{F}_1}h_i^{\prime}\pth{\tx}.
\end{align}
To see this, consider three possibilities:  (i) $\sum_{i\in \calR^*\cap\calF} h_i^{\prime}\pth{\tx}=0$,
 (ii) $\sum_{i\in \calR^*\cap\calF} h_i^{\prime}\pth{\tx}>0$, and (iii) $\sum_{i\in \calR^*\cap\calF} h_i^{\prime}\pth{\tx}<0$.

First consider the case when $\sum_{i\in \calR^*\cap\calF} h_i^{\prime}\pth{\tx}=0$. Due to  (\ref{boundary sign}) and the case assumption,
it holds that
$$\sum_{i\in \widetilde{F}_2}h_i^{\prime}\pth{\tx}\le 0=\sum_{i\in \calR^*\cap\calF} h_i^{\prime}\pth{\tx}=0\le \sum_{i\in \widetilde{F}_1}h_i^{\prime}\pth{\tx},$$
which is (\ref{sandwich}).

Now consider the case when  $\sum_{i\in \calR^*\cap\calF} h_i^{\prime}\pth{\tx}>0$. Since $\sum_{i\in \calR^*\cap\calF} h_i^{\prime}\pth{\tx}>0$, it follows that $\calR^*\cap \calF\not=\O$, and there exists $k\in \calR^*\cap \calF$ such that $h_k^{\prime}\pth{\tx}>0$. This implies that $h_i\pth{\tx}>0$ for each $i\in \widetilde{F}_1$.
Let $\mu=\min_{i\in \widetilde{F}_1} h_i^{\prime}(\tx)$. Note that $\mu>0$. By definition of $\widetilde{F}_1$, it follows that
$$h_i^{\prime}(\tx)\le \mu \le h_j^{\prime}(\tx),$$
for each $i\in \calR^*$ and $j\in \widetilde{F}_1$. Thus, we obtain
\begin{align}
\sum_{i\in \calR^*\cap\calF} h_i^{\prime}\pth{\tx}\le \sum_{i\in \calR^*\cap\calF} \mu=\pth{| \calR^*\cap\calF|}\mu
\le \pth{|\widetilde{F}_1|}\mu=\sum_{i\in \widetilde{F}_1} \mu \le \sum_{i\in \widetilde{F}_1} h_i^{\prime}\pth{\tx}.
\end{align}
Due to (\ref{boundary sign}) and the assumption that $\sum_{i\in \calR^*\cap \calF} h_i^{\prime}(\tx)>0$, we get
$$\sum_{i\in \widetilde{F}_2} h_i^{\prime}\pth{\tx}\le 0<\sum_{i\in \calR^*\cap\calF} h_i^{\prime}\pth{\tx}\le \sum_{i\in \widetilde{F}_1} h_i^{\prime}\pth{\tx},$$
proving relation (\ref{sandwich}).

Similarly, we can show the case when $\sum_{i\in \calR^*\cap\calF} h_i^{\prime}\pth{\tx}<0$.\\

Since the relation in (\ref{sandwich}) holds, there exists $0\le \zeta\le 1$ such that
\begin{align}
\label{gamma}
\sum_{i\in \calR^*\cap\calF} h_i^{\prime}\pth{\tx}=\zeta\pth{\sum_{i\in \widetilde{F}_1}h_i^{\prime}\pth{\tx}}+\pth{1-\zeta}\pth{\sum_{i\in \widetilde{F}_2}h_i^{\prime}\pth{\tx}}.
\end{align}
Thus, we have
\begin{align*}
0&=\sum_{i\in \calR^*}h_i^{\prime}\pth{\tx}  \text{~~~~~from (\ref{R bar})}\\
&=\sum_{i\in \calR^*-\calF} h_i^{\prime}\pth{\tx}+\sum_{i\in \calR^*\cap\calF} h_i^{\prime}\pth{\tx}\\
&=\sum_{i\in \calR^*-\calF} h_i^{\prime}\pth{\tx}+\zeta\pth{\sum_{i\in \widetilde{F}_1}h_i^{\prime}\pth{\tx}}+\pth{1-\zeta}\pth{\sum_{i\in \widetilde{F}_2}h_i^{\prime}\pth{\tx}}
\text{~~~~~using (\ref{gamma})}
\end{align*}
Thus $\tx$ is an optimum of function
\begin{align*}
\sum_{i\in \calR^*-\calF} h_i\pth{x}+\zeta\pth{\sum_{i\in \widetilde{F}_1}h_i\pth{x}}+\pth{1-\zeta}\pth{\sum_{i\in \widetilde{F}_2}h_i\pth{x}}
\end{align*}
Since constant scaling does not change optima, it follows that $\tx$ is an optimum of function
\begin{align}
\label{talg1-obj=0}
\frac{1}{|\calR^*-\calF|+\zeta |\widetilde{F}_1|+\pth{1-\zeta} |\widetilde{F}_2|}\pth{\sum_{i\in \calR^*-\calF}h_i\pth{x}+\zeta\sum_{i\in \widetilde{F}_1} h_i\pth{x}+\pth{1-\zeta} \sum_{i\in \widetilde{F}_2} h_i\pth{x}}.
\end{align}
Since $|\calR^*|=n-2f$ and $|\widetilde{F}_1|=f-\phi+|\calR^*\cap \calF|=|\widetilde{F}_2|$, we have
\begin{align*}
|\calR^*-\calF|+\zeta |\widetilde{F}_1|+\pth{1-\zeta} |\widetilde{F}_2|&=|\calR^*-\calF|+|\widetilde{F}_1|~~~\text{since}~ |\widetilde{F}_1|=|\widetilde{F}_2|\\
&=|\calR^*-\calF|+f-\phi+|\calR^*\cap \calF|\\
&=|\calR^*|-|\calR^*\cap \calF|+f-\phi+|\calR^*\cap \calF|\\
&=|\calR^*|+f-\phi=n-2f+f-\phi=n-\phi-f=|\calN|-f.
\end{align*}
We know that either $\zeta\ge \frac{1}{2}$ or $1-\zeta\ge \frac{1}{2}$; by symmetry, without loss of generality, assume $\zeta\ge \frac{1}{2}$. In addition, we know
\begin{align*}
|\pth{\calR^*-\calF}\cup \widetilde{F}_1|&=|\calR^*-\calF|+|\widetilde{F}_1|\\
&=|\calR^*|-|\calR^*\cap \calF|+f-\phi+|\calR^*\cap \calF|=|\calN|-f.
\end{align*}
Recall that $\calR^*\cup \bar{F}_1\cup \bar{F}_2=\calV$.
Thus, in function (\ref{talg1-obj=0}), which is a weighted sum of $|\calN|$ local cost functions
corresponding to agents in $\calN=\calV-\calF$, at least $|\calN|-f$ local cost functions corresponding to $i\in \pth{\calR^*-\calF}\cup \widetilde{F}_1$ have weights that are lower bounded by $\frac{1}{2(|\calN|-f)}$.\\

Similarly, when $1-\zeta\ge \frac{1}{2}$, we have at least $|\calN|-f$ cost functions corresponding to $i\in \pth{\calR^*-\calF}\cup \widetilde{F}_2$ have weight lower bounded by $\frac{1}{2(|\calN|-f)}$.

\eproof
\end{proof}

Algorithm 1 has the following alternative performance guarantee. The relative strength of Theorem \ref{t_algo1} and Theorem \ref{t_algo11} depends on the value of $n$. For instance, when $n\ge 4f+1$, it holds that $\frac{1}{n}\ge \frac{1}{2(|\calN|-f)}$.

\begin{theorem}
\label{t_algo11}
When $n\ge 3f+1$, Algorithm 1 solves Problem 3 with $\beta=\frac{1}{n}$ and $\gamma=|\calN|-f$.
\end{theorem}
\begin{proof}
From the proof of Theorem \ref{t_algo1}, we know that under Algorithm 1, (\ref{boundary sign}) holds, i.e.,
$$\sum_{i\in \widetilde{F}_1} h_i^{\prime}(\tx)\ge 0~~~\text{and}~~~\sum_{i\in \widetilde{F}_2} h_i^{\prime}(\tx)\le 0.$$
Then there exists $0\le \zeta\le 1$ such that
\begin{align}
\label{rewritten 0 alg1}
0~=~ \zeta\pth{\sum_{i\in \widetilde{F}_1} h_i^{\prime}(\tx)}+(1-\zeta)\pth{\sum_{i\in \widetilde{F}_2} h_i^{\prime}(\tx)}.
\end{align}
Note that either $\zeta\ge \frac{1}{2}$ or $1-\zeta\ge \frac{1}{2}$. Our proof focuses on the scenario when $\zeta\ge \frac{1}{2}$ -- the scenario when $1-\zeta \ge \frac{1}{2}$ can be proved analogously. \\

Recall from (\ref{R bar}) that
$$0~=~\sum_{i\in \calR^*} h_i^{\prime}(\tx)=\sum_{i\in \calR^*-\calF} h_i^{\prime}(\tx)+\sum_{i\in \calR^*\cap\calF} h_i^{\prime}(\tx).$$

We now consider two cases: (i) $\sum_{i\in \calR^*\cap\calF} h_i^{\prime}\pth{\tx}=0$,
and (ii)
$\sum_{i\in \calR^*\cap\calF} h_i^{\prime}\pth{\tx}\not=0$.

\paragraph{\bf Case (i):
$\sum_{i\in \calR^*\cap\calF} h_i^{\prime}\pth{\tx}=0$.} In this case, we have
\begin{align}
\label{alg 1-2-1}
\nonumber
0&=\sum_{i\in \calR^*}h_i^{\prime}\pth{\tx}  \text{~~~~~from (\ref{R bar})}\\
\nonumber
&=\sum_{i\in \calR^*-\calF} h_i^{\prime}\pth{\tx}+\sum_{i\in \calR^*\cap\calF} h_i^{\prime}\pth{\tx}\\
\nonumber
&=\sum_{i\in \calR^*-\calF} h_i^{\prime}\pth{\tx}+0 \text{~~~~due to the assumption in case (i)}\\
&=\sum_{i\in \calR^*-\calF} h_i^{\prime}\pth{\tx}.
\end{align}
Multiplying both sides of (\ref{alg 1-2-1}) by $\zeta$, we get
\begin{align*}
0~=~\zeta~ 0 &=\zeta\pth{\sum_{i\in \calR^*-\calF} h_i^{\prime}\pth{\tx}}\\
&=\zeta\pth{\sum_{i\in \calR^*-\calF} h_i^{\prime}\pth{\tx}}+0\\
&=\zeta\pth{\sum_{i\in \calR^*-\calF} h_i^{\prime}\pth{\tx}}+\zeta\pth{\sum_{i\in \widetilde{F}_1}h_i^{\prime}\pth{\tx}}+\pth{1-\zeta}\pth{\sum_{i\in \widetilde{F}_2}h_i^{\prime}\pth{\tx}} \text{~~~~~using (\ref{rewritten 0 alg1})}
\end{align*}
Thus $\tx$ is an optimum of function
\begin{align*}
\zeta\pth{\sum_{i\in \calR^*-\calF} h_i\pth{x}}+\zeta\pth{\sum_{i\in \widetilde{F}_1}h_i\pth{x}}+\pth{1-\zeta}\pth{\sum_{i\in \widetilde{F}_2}h_i\pth{x}}
\end{align*}
Since constant scaling does not change optima, it follows that $\tx$ is an optimum of function
\begin{align}
\label{obj=0}
\frac{1}{\zeta|\calR^*-\calF|+\zeta |\widetilde{F}_1|+\pth{1-\zeta} |\widetilde{F}_2|}\pth{\zeta\sum_{i\in \calR^*-\calF}h_i\pth{x}+\zeta\sum_{i\in \widetilde{F}_1} h_i\pth{x}+\pth{1-\zeta} \sum_{i\in \widetilde{F}_2} h_i\pth{x}}.
\end{align}
Since $\zeta\ge \frac{1}{2}\ge 1-\zeta$, $|\calR^*|=n-2f$ and $|\widetilde{F}_1|=f-\phi+|\calR^*\cap \calF|=|\widetilde{F}_2|$, we get
\begin{align*}
\zeta|\calR^*-\calF|+\zeta |\widetilde{F}_1|+\pth{1-\zeta} |\widetilde{F}_2|&\le
\zeta|\calR^*-\calF|+\zeta |\widetilde{F}_1|+\zeta |\widetilde{F}_2|~~~\text{since}~\zeta\ge \frac{1}{2}\ge 1-\zeta\\
&= \zeta \pth{|\calR^*-\calF|+|\widetilde{F}_1|+|\widetilde{F}_2|}\\
&=\zeta \pth{|\calR^*|-|\calR^*\cap\calF|+2f-2\phi+2|\calR^*\cap\calF|}\\
&=\zeta \pth{n-2f+2f-2\phi+|\calR^*\cap\calF|}\\
&\le \zeta \pth{n-2f+2f-2\phi+\phi}~~~\text{since}~|\calR^*\cap \calF|\le |\calF|=\phi\\
&=\zeta\pth{n-\phi}= \zeta |\calN| \le \zeta n.
\end{align*}
In addition, we know
\begin{align*}
 |\pth{\calR^*-\calF}\cup \widetilde{F}_1|&=|\calR^*-\calF|+|\widetilde{F}_1|\\
 &=|\calR^*|-|\calR^*\cap \calF|+f-\phi+|\calR^*\cap \calF|\\
 &=n-2f+f-\phi=n-\phi-f=|\calN|-f.
\end{align*}
Recall that $\calR^*\cup \bar{F}_1\cup \bar{F}_2=\calV$.
Thus, in function (\ref{obj=0}), which is a weighted sum of $|\calN|$ local cost functions
corresponding to agents in $\calN=\calV-\calF$, at least $|\calN|-f$ local cost functions corresponding to $i\in \pth{\calR^*-\calF}\cup \widetilde{F}_1$ have weights that are lower bounded by $\frac{1}{n}$.\\

Similarly, when $1-\zeta\ge \frac{1}{2}$, we have at least $|\calN|-f$ cost functions corresponding to $i\in \pth{\calR^*-\calF}\cup \widetilde{F}_2$ have weight lower bounded by $\frac{1}{n}$.

\paragraph{\bf Case (ii): $\sum_{i\in \calR^*\cap\calF} h_i^{\prime}\pth{\tx}\not=0$.}
By symmetry, without loss of generality, assume that $\sum_{i\in \calR^*\cap\calF} h_i^{\prime}\pth{\tx}>0$. Then $\calR^*\cap \calF\not=\O$ and there exists $j\in \calR^*\cap \calF$ such that $h_j^{\prime}\pth{\tx}>0$. This implies that $F_1^*=\bar{F}_1$. 

By (\ref{sandwich}), we get
$$0~<~ \sum_{i\in \calR^*\cap \calF} h_i^{\prime}(\tx)\le \sum_{i\in \widetilde{F}_1} h_i^{\prime}(\tx).$$
Thus, there exists $0\le \zeta_1\le 1$ such that
\begin{align}
\label{alg case 2}
\sum_{i\in \calR^*\cap \calF} h_i^{\prime}(\tx)=\zeta_1 \pth{\sum_{i\in \widetilde{F}_1} h_i^{\prime}(\tx)}
\end{align}
By equation (\ref{R bar}), we have
\begin{align}
\label{alg 1-2-2}
\nonumber
0~&=\sum_{i\in \calR^*-\calF} h_i^{\prime}(\tx)+\sum_{i\in \calR^*\cap \calF} h_i^{\prime}(\tx)\\
&=\sum_{i\in \calR^*-\calF} h_i^{\prime}(\tx)+\zeta_1 \pth{\sum_{i\in \widetilde{F}_1} h_i^{\prime}(\tx)}~~~\text{by}~(\ref{alg case 2}).
\end{align}
Multiplying both sides of (\ref{alg 1-2-2}) by $\zeta$, we get
\begin{align*}
0=\zeta\, 0&=\zeta\pth{\sum_{i\in \calR^*-\calF} h_i^{\prime}(\tx)+\zeta_1 \sum_{i\in \widetilde{F}_1} h_i^{\prime}(\tx) }\\
&=\zeta\pth{\sum_{i\in \calR^*-\calF} h_i^{\prime}(\tx)+\zeta_1 \sum_{i\in \widetilde{F}_1} h_i^{\prime}(\tx)}+0\\
&=\zeta\pth{\sum_{i\in \calR^*-\calF} h_i^{\prime}(\tx)+\zeta_1 \sum_{i\in \widetilde{F}_1} h_i^{\prime}(\tx)}+\zeta\pth{\sum_{i\in \widetilde{F}_1}h_i^{\prime}\pth{\tx}}+\pth{1-\zeta}\pth{\sum_{i\in \widetilde{F}_2}h_i^{\prime}\pth{\tx}}~~~\text{by}~~ (\ref{rewritten 0 alg1})
\end{align*}
Thus $\tx$ is an optimum of function
\begin{align*}
\zeta\pth{\sum_{i\in \calR^*-\calF} h_i(\tx)+\zeta_1 \sum_{i\in \widetilde{F}_1} h_i(\tx)}+\zeta\pth{\sum_{i\in \widetilde{F}_1}h_i\pth{\tx}}+\pth{1-\zeta}\pth{\sum_{i\in \widetilde{F}_2}h_i\pth{\tx}}
\end{align*}
Define $\chi=\zeta\pth{|\calR^*-\calF|+\zeta_1 |\widetilde{F}_1|}+\zeta |\widetilde{F}_1|+ \pth{1-\zeta} |\widetilde{F}_2|$.
Since constant scaling does not change optima, we know that $\tx$ is an optimum of function
\begin{align}
\label{objnot=0}
\frac{1}{\chi}
\pth{\zeta \pth{\sum_{i\in \calR^*-\calF} h_i(\tx)+\zeta_1 \sum_{i\in \widetilde{F}_1} h_i(\tx)}+\zeta \sum_{i\in \widetilde{F}_1}h_i\pth{\tx}+\pth{1-\zeta}\sum_{i\in \widetilde{F}_2}h_i\pth{\tx}
}.
\end{align}
Since $\calR^*\cap \calF\not=\O$, it holds that $|\widetilde{F}_1|=f-\phi+|\calR^*\cap \calF|\ge |\calR^*\cap \calF|> 0$, i.e., $ \widetilde{F}_1\not=\O$.
Define $$\min_{j\in \widetilde{F}_1} h_i^{\prime}(\tx)= \mu.$$ By definition of $\mu$, it holds that $\mu\le h_i^{\prime}(\tx)$ for $i\in \widetilde{F}_1$. Also, $\mu>0$ because  $h_i^{\prime}(\tx)>0$ for $i\in \widetilde{F}_1\subseteq \bar{F}_1=F_1^*\subseteq A(\tx)$. In addition, by the construction of $\calR^*$, we have $h_i^{\prime}\pth{\tx}\le \mu$ for each $i\in \calR^*\cap \calF$.
Then from (\ref{alg case 2})
\begin{align}
\label{zeta1}
\zeta_1|\widetilde{F}_1|&=\frac{\sum_{i\in \calR^*\cap\calF} h_i^{\prime}\pth{\tx}}{\sum_{i\in \widetilde{F}_1} h_i^{\prime}\pth{\tx}}|\widetilde{F}_1|\nonumber\\
&\le \frac{\sum_{i\in \calR^*\cap\calF}\mu}{\sum_{i\in \widetilde{F}_1} \mu}|\widetilde{F}_1|\nonumber\\
&=\frac{|\calR^*\cap\calF|}{|\widetilde{F}_1|}|\widetilde{F}_1|=|\calR^*\cap\calF|.
\end{align}
Then
\begin{align*}
\chi~&=~\zeta\pth{|\calR^*-\calF|+\zeta_1 |\widetilde{F}_1|}+\zeta |\widetilde{F}_1|+ \pth{1-\zeta} |\widetilde{F}_2|\\
&\le ~\zeta\pth{|\calR^*-\calF|+|\calR^*\cap\calF|}+\zeta |\widetilde{F}_1|+ \pth{1-\zeta} |\widetilde{F}_2|~~~\text{by}~(\ref{zeta1})\\
&=~\zeta|\calR^*|+\zeta |\widetilde{F}_1|+ \zeta |\widetilde{F}_2|~~~\text{since}~~~\zeta\ge \frac{1}{2}\ge 1-\zeta\\
&=~\zeta\pth{|\calR^*|+|\widetilde{F}_1|+|\widetilde{F}_2|}\\
&=~\zeta\pth{n-2f+2f-2\phi+2|\calR^*\cap \calF|}\\
&\le~ \zeta\pth{n-2f+2f-2\phi+2\phi}~~~\text{since}~~~|\calR^*\cap \calF|\le |\calF|=\phi\\
&=~\zeta\, n.
\end{align*}
In addition, we know
\begin{align*}
|\pth{\calR^*-\calF}\cup \widetilde{F}_1|&=|\calR^*-\calF|+|\widetilde{F}_1|\\
&=|\calR^*|-|\calR^*\cap\calF|+|\widetilde{F}_1|\\
~&=n-2f-|\calR^*\cap\calF|+f-\phi+|\calR^*\cap\calF|\\
~&=n-\phi-f = |\calN|-f.
\end{align*}
Thus, in function (\ref{objnot=0}), at least $|\calN|-f$ local cost functions corresponding to $i\in \pth{\calR^*-\calF}\cup \widetilde{F}_1$ are assigned with weights that are lower bounded by $\frac{1}{n}$. Similar result holds when $1-\zeta\ge \frac{1}{2}$.


Cases (i) and (ii) together prove the theorem.

\eproof
\end{proof}

Recall that $H\pth{\cdot}=F\pth{\cdot}+G\pth{\cdot}$.
We will use the result below later to prove correctness of Algorithm 3. Let $Cov\pth{\cdot}$ be the convex hull of a given set.


\begin{theorem}
\label{globalvirtualfunction}
For given $\calN$ and $\calF$, there exists a convex and differentiable function ${\bf H}\pth{\cdot}$ defined over
any finite interval $[c,d]\supseteq Cov\pth{\cup_{i\in \calN}X_i}$ such that the derivative function of ${\bf H}\pth{\cdot}$ is $H\pth{\cdot}$, i.e., ${\bf H}^{\prime}\pth{x}=H\pth{x}$ for each $x\in[c,d]$
where $Cov\pth{\cup_{i\in \calN}X_i}\subseteq [c,d]$.
\end{theorem}
\begin{proof}
We prove the existence of ${\bf H}\pth{\cdot}$ by construction. Specifically we show that function $H\pth{\cdot}$ is integrable over $Cov\pth{\cup_{i\in \calN}X_i}$.

By definition of admissible functions, $X_i=\arg\min_{x\in\mathbb{R}} h_i(x)$ is nonempty and compact (closed and bounded).
$Cov\pth{\cup_{i\in \calN}X_i}$ is the convex hull spanned by the union of $X_i$'s for all $i\in \calN$. Thus $Cov\pth{\cup_{i\in \calN}X_i}$ is convex and compact. In addition, by Propositions \ref{monotone} and \ref{continuous}, we know that function $H\pth{\cdot}=F\pth{\cdot}+G\pth{\cdot}$ is non-decreasing and continuous.

As stated in the theorem,
$Cov\pth{\cup_{i\in \calN}X_i}\subseteq [c,d]$.
 Let $H\pth{\cdot}_{[c,d]}$ be the restriction of function $H\pth{\cdot}$ to the closed interval $[c,d]$. Then we know that $H\pth{\cdot}_{[c,d]}$ is Riemann integrable over the closed interval
$[c,d]$.

For $x\in [c,d]$, define ${\bf H}\pth{x}$ by
\begin{align*}
{\bf H}\pth{x}\triangleq \int_{c}^x H\pth{t}dt.
\end{align*}
Since $H\pth{\cdot}$ is continuous on $[c, d]$,  we know ${\bf H}(\cdot)$ is differentiable and
\begin{align*}
{\bf H}^{\prime}(x)=H(x),
\end{align*}
for all $x\in (c,d)$ \cite{ross1980elementary}.

In addition, by the fact that a scalar differentiable function is convex on an interval if and only if its derivative is non-decreasing on that interval, we know that ${\bf H}\pth{\cdot}$ is convex.

\eproof
\end{proof}
~

It is easy to see that the function ${\bf H}(\cdot)$, defined in Theorem \ref{globalvirtualfunction}, is $nL$--Lipschitz continuous on any finite interval $[c,d]\supseteq Cov\pth{\cup_{i\in \calN}X_i}$.

\begin{remark}
The correctness of Algorithm 1 implies that $H\pth{\tx}=0$ and $\tx\in Cov\pth{\cup_{i\in \calN}X_i}$, where the latter claim follows from Proposition \ref{p1}, proved in Appendix \ref{appendix:p1}. Essentially, Algorithm 1 outputs an optimum of the following constrained convex optimization problem, where $Cov\pth{\cup_{i\in \calN}X_i}\subseteq [c,d]$:
\begin{align}
\label{equivalence of Algorithm 1}
\min \quad&{\bf H}\pth{x}\\
s.t. \quad& x\in [c,d]\nonumber.
\end{align}
\end{remark}

%

\subsection{Algorithm $2$}

Algorithm 2 is an alternative to Algorithm 1. This construction admits more concise proofs. The Step 1 in Algorithm 2 is identical to that in Algorithm 1. What distinguishes Algorithm 2 from Algorithm 1 is the decision rule described in Step 2.

\paragraph{}
\vspace*{8pt}\hrule
~

{\bf Algorithm $2$} for agent $j$:
\vspace*{4pt}\hrule

\begin{list}{}{}
\item[{\bf Step 1:}]
Perform Byzantine broadcast of local cost function $h_j(x)$ to all the agents
using any Byzantine broadcast algorithm, such as \cite{psl_BG_1982}.

In step 1,
agent $j$ should receive from each agent $i\in\calV$ its cost function $h_i(x)$.
For non-faulty agent $i\in\calN$, $h_i(x)$ will be an admissible function
({\em admissible} is defined in Section \ref{sec:intro}).
If a faulty agent $k\in\calF$ does not correctly perform Byzantine broadcast of
its cost function, or broadcasts
an inadmissible cost function, then hereafter assume $h_k(x)$
to be a {\em default} admissible cost function that is known to all agents.\\

\item[{\bf Step $2$:}]
The multiset of admissible functions obtained in Step 1 is
$\{h_1(x), h_2(x), \cdots, h_n(x)\}$ and the multiset of the derivatives of these functions is $\{h_1^{\prime}(x), h_2^{\prime}(x), \cdots, h_n^{\prime}(x)\}$.

For $1\le K\le n$, define
\begin{align}
\label{def of rankfunction}
g_K(x)\triangleq K^{th} ~\text{largest value (including ties) in the multiset}~ \{h_1^{\prime}(x), h_2^{\prime}(x), \cdots, h_n^{\prime}(x)\},
\end{align}
for each $x\in \reals$. If there exists $x\in\mathbb{R}$ such that
\begin{align}
\label{alg1.3}
\sum_{K=f+1}^{n-f} g_K\pth{x}~=~0~~~~
\end{align}
then deterministically choose output $\tx$ to be any one $x$ value that satisfies (\ref{alg1.3});\\otherwise, choose output $\tx=\perp$.
\end{list}

~
\hrule

~
\\

Similar to Algorithm 1, we can also show that Algorithm $2$ solves Problem 3 with parameters $\beta=\frac{1}{2(|\calN|-f)}$ and $\gamma=|\calN|-f$. Our correctness proof is based on the following fact.

\begin{proposition}
\label{Continous and Monotone of the K-th Rank Functions}
For each $1\le K\le n$, the function $g_K(x)$ defined in (\ref{def of rankfunction}) is a continuous non-decreasing function.
\end{proposition}

The proof of Proposition \ref{Continous and Monotone of the K-th Rank Functions} is presented in Appendix \ref{Proof of Continous and Monotone of the K-th Rank Functions}.

\begin{lemma}
\label{alternative well-definedOutPut}
Algorithm 2 returns $\tx\in\mathbb{R}$ when $n>3f$ (i.e., it does not return $\bot$).
\end{lemma}

The proof of Lemma \ref{alternative well-definedOutPut} is similar to the proof of Lemma \ref{well-definedOutPut}. We present it in Appendix \ref{appendix: alternative well-definedOutPut}.

The next two theorems prove that Algorithm 2 can solve Problem 3 for
$\gamma=|\calN|-f$.

\begin{theorem}
\label{t_algo2}
When $n>3f$,
Algorithm $2$ solves Problem 3 with $\beta=\frac{1}{2(|\calN|-f)}$ and $\gamma=|\calN|-f$.
\end{theorem}
\begin{theorem}
\label{t_algo2-1}
When $n>3f$,
Algorithm $2$ solves Problem 3 with $\beta=\frac{1}{n}$ and $\gamma=|\calN|-f$.
\end{theorem}

The proofs of Theorems \ref{t_algo2} and \ref{t_algo2-1} are presented in Appendix \ref{appendix:t_algo2}.

\subsection{Algorithm 3}
\label{algo3}
Both Algorithm 1 and Algorithm $2$ may be impractical because they require entire cost functions to be exchanged between the agents.
Unlike Algorithms 1 and 2, the iterative Algorithm 3 presented below
does {\em not} require the agents to exchange their local cost functions
in their entirety. Instead, the agents exchange gradients of their local
cost functions in each iteration. This algorithm, derived
from the gradient-based methods for convex optimization,  is more suitable for practical purposes.\\

For convergence of our gradient-based algorithm (Algorithm 3),
we now impose the additional restriction that
the admissible functions also be $L$-Lipschitz continuous. 
This restriction is assumed in the rest of this section, even if it is not
stated explicitly again.

In Algorithm 3, each agent $j$ computes a state variable
$x_j[t]$ in the $t$-th iteration, $t\geq 1$, as elaborated below.
We assume that the entire sequence $x_j[t]$ is saved by agent $j$. Due to this requirement, Algorithm 3 requires additional memory, rather than keeping the minimal state $x_i[t]$. In the next subsection, we reduce this state maintenance requirement.

\paragraph{}
\hrule
~
\vspace*{4pt}

{\bf Algorithm 3} for agent $j$:
~
\vspace*{4pt}\hrule

\begin{list}{}{}
\item[{\bf Initialization Step (i):}]
Choose $v_j\in X_j=\arg\min_{x\in\mathbb{R}} h_j(x)$.\\
~
\item[{\bf Initialization Step (ii):}] Perform exact Byzantine consensus with $v_j$ as
the input of agent $j$ to the consensus algorithm -- any exact Byzantine
consensus algorithm \cite{Lynch:1996:DA:525656} may be used for this purpose\footnote{For instance, consider the following algorithm: Each agent Byzantine broadcasts its $v_j$; Collect $n$ such values from other agents; Drop the smallest $f$ values and the largest $f$. The output or consensus is the average of the remaining $n-2f$ values. }.
Set $x_j[0]$ to the output of the above consensus algorithm.\\

\item[{\bf Iteration $t\geq 1$:}]

\begin{itemize}
\item[{\bf Step 1:}] Compute $h_j^{\prime}\pth{x_j[t-1]}$, and perform Byzantine broadcast of $h_j^{\prime}\pth{x_j[t-1]}$ to all the agents, using any Byzantine broadcast algorithm, such as \cite{psl_BG_1982}.\\

In step 1, agent $j$ should receive a gradient from each agent $i\in\calV$ -- let us denote the gradient received from agent $i$ in iteration $t$ as $g_i[t-1]$.
Agent $j$ keeps a record of the sequence $(t, x_j[t])$.
Agent $j$ also keeps a record of the sequence $(t,g_i[t-1])$ for each agent $i$.
Each received gradient is checked for {\em admissibility} as follows, using
the above record.\\
\begin{itemize}
\item If no gradient is, in fact, received from agent $i$ in iteration $t$
via a Byzantine broadcast from $i$,
then the gradient $g_i[t-1]$ for agent $i$ is deemed {\em inadmissible}.

\item
If there exists an iteration $1\leq t_0<t$ such that at least one of the following conditions is true, then the gradient received from agent $i$ is deemed
{\em inadmissible}.\\
\begin{enumerate}
\item $x_j[t_0-1]\le x_j[t-1]$ and $g_k[t_0-1]>g_k[t-1]$ \\
\item $x_j[t_0-1]\ge x_k[t-1]$ and $g_k[t_0-1]<g_k[t-1]$ \\
\item $|g_k[t-1]|> L$
\end{enumerate}
\end{itemize}
~

If the gradient received from any agent $i$ is deemed inadmissible, then
it must be the case that agent $i$ is faulty. In that case, agent $i$ is isolated (i.e.,
removed from the system). This reduces the total number of agents $n$ by 1,
and the maximum number of faulty agent $f$ is also reduced by 1.
Algorithm 3 is {\em restarted} (from Step 1) using the new parameters
$n$ and $f$.\footnote{It is also possible to continue executing the
algorithm further, but for brevity, we take the approach of eliminating
the faulty agent, and restarting.}
The gradients received from
any non-faulty agent $i\in\calN$ will never be found to be inadmissible.\\

\item[{\bf Setp 2:}] Due to the restart mechanism above,
the algorithm progresses to Step 2 only when all the received gradients
are deemed admissible.
Let $\calR[t-1]$ denote the multiset of admissible gradients $\{g_1[t-1],g_2[t-1],\cdots,g_n[t-1]\}$, obtained in Step 1 of $t$-th iteration.
%
%
If there are more than $f$ positive gradients in $\calR[t-1]$, then remove $f$ largest gradients from $\calR[t-1]$; otherwise, remove all positive gradients from $\calR$.
Further, if there are more than $f$ negative gradients in $\calR[t-1]$, then remove $f$ smallest gradients from $\calR$; otherwise, remove all negative gradients from $\calR[t-1]$.
Let $\calR^*[t-1]$ be the set of agents corresponding to all the remaining gradients.
\\

\item[{\bf Step 3:}] Let $\{\lambda[t]\}_{t=0}^{\infty}$ be a sequence of diminishing (non-increasing and $\lambda[t]\to 0$) stepsizes chosen beforehand such that $\lambda[t]>0$ for each $t$, $\sum_{t=1}^{\infty}\lambda[t]=\infty$ and $\sum_{t=1}^{\infty}\lambda^2[t]<\infty$.\footnote{For instance, the stepsizes $\{\lambda[t]=\frac{1}{t+1}\}_{t=0}^{\infty}$ satisfy the aforementioned condition, since $\lambda[t]=\frac{1}{t+1}\le \frac{1}{t+2}=\lambda[t+1]$, $\sum_{t=0}^{\infty} \frac{1}{t+1}=\infty$ and $\sum_{t=0}^{\infty} \frac{1}{\pth{t+1}^2}<\infty$. }

Compute $x_j[t]$ as
\begin{align*}
x_j[t]= x_j[t-1]-\lambda[t-1]\sum_{i\in \calR^*[t-1]} g_i[t-1].
\end{align*}

\end{itemize}

\end{list}

~
\hrule

~

We prove that Algorithm 3 solves Problem 3 with $\gamma=|\calN|-f$. In fact, we will see later that Algorithm 3 is essentially the gradient method for the constrained convex optimization problem (\ref{equivalence of Algorithm 1}).  Then the remaining correctness proof follows the standard convergence analysis of such algorithms \cite{boyd2004convex,nesterov2004introductory}.

\begin{lemma}
\label{initial}
The consensus value obtained at the end of step 2 of Algorithm 3 is contained in $Cov\pth{\cup_{i\in\calN}X_i}$, i.e., $x[0]\in Cov\pth{\cup_{i\in\calN}X_i}$.
\end{lemma}
The above lemma follows trivially from {\em validity} condition
imposed on a correct Byzantine consensus algorithm. Thus the proof is omitted
here. \\

Let $[a, b]\triangleq Cov\pth{\cup_{i\in \calN}X_i}$. Recall that all functions $h_i(x)$'s are $L$-Lipschitz continuous, i.e., $|h_i^{\prime}(x)|\le L$.


\begin{proposition}
\label{trapped}
In Algorithm 3,  $x_i[t]=x_j[t]$ for all $i, j\in \calN$ and for all $t$. In addition, $x_i[t]\in [a-n\lambda[0] L, b+n\lambda[0] L]$ for any $t$.
\end{proposition}
\begin{proof}
Recall that admissible cost functions are $L$-Lipschitz.
We prove this proposition by induction. By the correctness of a Byzantine consensus algorithm, $x_i[0]=x_j[0]$ for all $i,j\in \calN$. Assume that $x_i[t-1]=x_j[t-1]$ for all $i,j\in \calN$ for some $t>0$. At iteration $t$,
non-faulty agents $i$ and $j$ receive identical set of gradients $\calR[t-1]$ in Step 1
of iteration $t$. This, along with the assumption
that $x_i[t-1]=x_j[t-1]$ implies that $x_i[t]=x_j[t]$. \\

Now, by induction,
we show that $x_i[t]\in [a-n\lambda[0] L, b+n\lambda[0] L]$ for any $i\in \calN$ and for any $t\geq 0$. By the validity condition of a correct consensus algorithm, it holds that $x_i[0]\in Cov\pth{\cup_{i\in \calN} X_i}=[a,b]$. Assume that $x_i[t-1]\in [a-n\lambda[0] L, b+n\lambda[0] L]$ for some $t>0$. We will show that
\begin{align}
\label{update trapped}
x_i[t]\in [a-n\lambda[0] L, b+n\lambda[0] L]
\end{align}
is also true.
We know $[a-n\lambda[0] L, b+n\lambda[0] L]=[a-n\lambda[0] L, a)\cup [a, b]\cup (b, b+n\lambda[0] L]$. \\

When $x_i[t-1]\in [a, b]=Cov\pth{\cup_{i\in \calN} X_i}$, then
\begin{align*}
x_i[t]= x_i[t-1]-\lambda[t-1]\sum_{i\in \calR^*[t-1]} g_i[t-1]\overset{(a)}{\le} x_i[t-1]+n\lambda[0] L\le b+n\lambda[0] L,
\end{align*}
where $(a)$ follows due to the admissibility test (3) in Step 1 above and because $\lambda[0]\ge \lambda[t]$ for $t\ge 0$.

Similarly, we have
\begin{align*}
x_i[t]\ge a-n\lambda[0] L.
\end{align*}
Thus when $x_i[t-1]\in [a, b]$, then $x_i[t]\in [a-n\lambda[0] L, b+n\lambda[0] L]$. \\

When $x_i[t-1]\in (b, b+n\lambda[0] L]$, all non-faulty gradients are positive since $x_i[t-1]>b=\max_{j\in \calN}\max X_j$. Thus, at most $f$ admissible gradients are non-positive. 
By the code in Algorithm 3, all the negative admissible gradients will be removed. Thus $g_j[t-1]\ge 0$ for each $j\in \calR^*_i[t-1]$.
\begin{align*}
x_i[t]&= x_i[t-1]-\lambda[t-1]\sum_{i\in \calR^*[t-1]} g_i[t-1]\le x_i[t-1]\le b+n\lambda[0] L
\end{align*}
In addition,
\begin{align*}
x_i[t]&= x_i[t-1]-\lambda[t-1]\sum_{i\in \calR^*[t-1]} g_i[t-1] \ge x_i[t-1]-n\lambda[0] L\ge a-n\lambda[0] L.
\end{align*}
because $x_i[t-1]>b\geq a$ and $g_i[t-1]\leq L$
for $g_i[t-1]\in \calR^*[t-1]$.

Thus when $x_i[t-1]\in (b, b+n\lambda[0] L]$, $x_i[t]\in [a-n\lambda[0] L, b+n\lambda[0] L]$. Similarly, we can show that
when $x_i[t-1]\in  [a-n\lambda[0] L, a)$, $x_i[t]\in [a-n\lambda[0] L, b+n\lambda[0] L]$.

Therefore, we conclude that $x_i[t]\in [a-n\lambda[0] L, b+n\lambda[0] L]$ and the induction is complete.
\eproof
\end{proof}

\vskip 1.5\baselineskip
Henceforth, with an abuse of notation, we drop the subscript $j$ of $x_j[t]$ for each $j$ and $t$. Similarly, we drop the time index $[0]$ of $\lambda[0]$.

\begin{theorem}
\label{localvirtualfunction}
For any $i\in \calF$, let $\{g_i[t-1])\}_{t=1}^{\infty}$ be the sequence of admissible gradients generated by Algorithm 3, where
$g_i[t-1]$ is supposed to be the gradient at $x[t-1]$. Then there exists a function $g(x)$ defined over $[c, d]$, which contains points $a-n\lambda L$ and $b+n\lambda L$ as interior points, such that (i) $g^{\prime}(x[t-1])= g_i[t-1]$, and (ii) $g(x)$ is convex, $L$--Lipschitz, and differentiable.
\end{theorem}

\begin{proof}
When $i\in \calF$, let $g^{+}=\sup~ \{g_i[t-1])\}_{t=1}^{\infty}$, $g^{-}=\inf~ \{g_i[t-1])\}_{t=1}^{\infty}$, and consider the piece-wise linear function with $\{\pth{x[t-1], g_i[t-1]}\}_{t=1}^{\infty}\cup \{\pth{c, ~g^{-}},\pth{d,~ g^{+}}\}$ as corners, denoted by $\widetilde {g}(x)$. It is easy to see that function $\widetilde {g}(x)$ is continuous.
Recall that $[c,d]$ is a closed interval. Thus $\widetilde {g}(x)$ is Riemann integrable over $[c,d]$. Choose $g(x)$ to be the integral function of $\widetilde {g}(x)$, i.e., $g(x)\triangleq \int_{c}^{x} \widetilde {g}(t)dt$. 

Since $\widetilde {g}(t)$ is continuous over $[c,d]$, it holds that $g(x)$ is differentiable and $g^{\prime}(x)=\widetilde {g}(x)$ for $x\in (c,d)$. In addition, when $\{g_i[t-1]\}_{t=1}^{\infty}$ are admissible, by definition of admissible gradients (in Algorithm 3), function $\widetilde {g}(x)$ is non-decreasing.
Then function $g(x)$ is convex over $[c,d]$. 
%
$L$-Lipschitz property is also ensured by admissibility of the gradients.
\eproof
\end{proof}

It is easy to see that there exists an admissible function $\bar{g}(\cdot)$ such that the restriction of $\bar{g}(\cdot)$ to $[c,d]$ equals $g(\cdot)$, i.e., $\bar{g}|_{[c,d]}(\cdot)=g(\cdot)$. Thus, Theorem \ref{localvirtualfunction} states that the
record of gradients saved in Algorithm 3 effectively
forces each Byzantine agent (that is not already isolated)
to behave as if it is non-faulty with a local cost function $\bar{g}(x)$ that is admissible.
Therefore, hereafter we can assume that all agents, including faulty
agents, behave correctly and consistently with an admissible local
cost function.  \\

Recall that we defined
$[a,b]= Cov\pth{\cup_{i\in \calN} X_i}$. By Proposition \ref{trapped}, we know that the local estimate of each non-faulty agent $i$ is trapped within the closed interval $[a-n\lambda[0] L, b+n\lambda[0] L]$ for all iterations, i.e., $x_i[t]\in [a-n\lambda[0] L, b+n\lambda[0] L]$ for all $i\in \calN$ and all $t$. Therefore, Algorithm 3 is essentially trying to find an  (exact or approximate)
optimum of the following constrained convex optimization problem, which is a variant of (\ref{equivalence of Algorithm 1}):
\begin{align*}
\min \quad &{\bf H}\pth{x}\\
s.t. \quad & x\in [a-nL\lambda[0], b+nL\lambda[0]].
\end{align*}

It should be easy to see that the total gradient $\sum_{i\in \calR^*[t-1]} g_i[t-1]$
used in computing $x_j[t]$ is identical to $F(x_j[t-1])+G(x_j[t-1])$,
which is the gradient of ${\bf H}(\cdot)$ at $x_j[t-1]$.
In other words, the agents are distributedly using the gradient method for convex optimization of global cost function ${\bf H}(\cdot)$, which is convex and continuous.

Following the convergence analysis of the gradient method in Theorem 3.2.2 in \cite{nesterov2004introductory} and Theorem 41 in \cite{NedicNotes}, we can show that the limit of $\{x[t]\}_{t=0}^{\infty}$ exists and $\lim_{t\diverge} x[t]=x^*$, where $x^*$ is an optimum of function ${\bf H}(\cdot)$.

\begin{remark}
The gradient-trimming mechanism in Step 2 of Algorithm 3  can be replaced by the following trimming rule:
``Remove $f$ largest gradients from $\calR[t-1]$ and remove $f$ smallest gradients from $\calR[t-1]$."
This trimming rule leaves $n-2f$ admissible gradients. The modified algorithm is then a distributed version of Algorithm 2 and its correctness can be shown analogously to that of Algorithm 3.
\end{remark}

\subsection{Algorithm 4: Non-Interleaved Iterative Algorithm I}
Recall that in Algorithm 3, the entire sequence $x_j[t]$ is saved by agent $j$. In the following two subsections, we relax this memory requirement by proposing two non-interleaved algorithms, in particular, Algorithms 4 and 5. Unlike Algorithms 1, 2, 3 and 5, the performance guarantee of Algorithm 4 is independent of $|\calN|$. In general, the performance of Algorithm 4 is weaker than that of Algorithms 1, 2, 3 and 5.

\paragraph{}
\hrule
~
\vspace*{4pt}

{\bf Algorithm 4} for agent $j$:
~
\vspace*{4pt}\hrule

\begin{list}{}{}
\item[{\bf Initialization Step (i):}]
Choose $x_j[0]\in X_j=\arg\min_{x\in\mathbb{R}} h_j(x)$.\\
~
\item[{\bf Initialization Step (ii):}] Perform exact Byzantine consensus with $x_j[0]$ as the input of agent $j$. Set $x[0]$ to the output of the above consensus algorithm.\\

\item[{\bf Iteration $t\geq 1$:}]

\begin{itemize}
\item[{\bf Step 1:}] Compute $h_j^{\prime}\pth{x_j[t-1]}$, and perform Byzantine broadcast of $h_j^{\prime}\pth{x_j[t-1]}$ to all the agents, using any Byzantine broadcast algorithm, such as \cite{psl_BG_1982}.\\

    In step 1, agent $j$ should receive a gradient from each agent $i\in\calV$ -- let us denote the gradient received from agent $i$ in iteration $t$ as $g_i[t-1]$.\footnote{For each $i\in \calN$, $g_i[t-1]=h_i^{\prime}(x[t-1])$.} If no gradient is, in fact, received from agent $i$ in iteration $t$
via a Byzantine broadcast from $i$, then it must be the case that agent $i$ is faulty. In that case,
 agent $i$ is isolated (i.e., removed from the system).\footnote{Alternatively, the gradient value can be replaced by some default value (say 0).}
 This reduces the total number of agents $n$ by 1,
and the maximum number of faulty agent $f$ is also reduced by 1.
Algorithm 4 is {\em restarted} (from Step 1) using the new parameters
$n$ and $f$.\\

\item[{\bf Setp 2:}] Due to the restart mechanism above,
the algorithm progresses to Step 2 only when agent $j$ has received gradients from all other nodes.
Let $\calR[t-1]$ denote the multiset of the gradients $\{g_1[t-1],g_2[t-1],\cdots,g_n[t-1]\}$, obtained in Step 1 of $t$-th iteration. Drop the smallest $f$ values and the largest $f$.
    Set $g[t-1]$ to the average of the remaining $n-2f$ values.\\

\item[{\bf Setp 3:}] Let $\{\lambda[t]\}_{t=0}^{\infty}$ be a sequence of diminishing (non-increasing and $\lambda[t]\to 0$) stepsizes chosen beforehand such that $\lambda[t]>0$ for each $t$, $\sum_{t=1}^{\infty}\lambda[t]=\infty$ and $\sum_{t=1}^{\infty}\lambda^2[t]<\infty$.

Compute $x[t]$ as
\begin{align}
\label{update valid non-iter1}
x[t]= x[t-1]-\lambda[t-1] g[t-1].
\end{align}

\end{itemize}

\end{list}

~
\hrule

~

~

Since $x[0]$ is the consensus value of the input $x_j[0]$ for each $j\in \calN$, and $g[0]$ is the consensus value of the non-faulty gradient $h_j^{\prime}(x[0])$, for each $j\in \calN$, the update function (\ref{update valid non-iter1}) is well-defined for $t=1$. By an inductive argument, we can show that the update (\ref{update valid non-iter1}) is well-defined for all $t$.\\

Let $\calC$ be the collection of functions defined as follows:
\begin{align}
\nonumber
\calC\triangleq \{~~~p(x): p(x)&=\sum_{i\in \calN} \alpha_i h_i(x), ~~\forall i\in\calN, ~ \alpha_i\geq 0,\\
\nonumber
&\sum_{i\in \calN}\alpha_i=1,\text{~~and~~}\\
&\sum_{i\in\calN} {\bf 1}\left(\alpha_i\ge \frac{1}{2(n-f)}\right) ~ \geq ~ n-2f ~~~ \}
\label{valid collection}
\end{align}
Each $p(x)\in \calC$ is called a valid function. Note that the function $\frac{1}{|\calN|}\sum_{i\in \calN}h_i(x)\in \calC$. For ease of future reference, we let $\widetilde{p}(x)=\frac{1}{|\calN|}\sum_{i\in \calN}h_i(x)$. Define $Y\triangleq \cup_{p(x)\in \calC} \argmin~ p(x)$.

\begin{lemma}
\label{valid hull non-iter1}
$Y$ is a convex set.
\end{lemma}
\begin{proof}
Let $x_1, x_2\in Y$ such that $x_1\not=x_2$. By definition of $Y$, there exist valid functions $p_1(x)=\sum_{i\in \calN}\alpha_i h_i(x)\in \calC$ and $p_2(x)=\sum_{i\in \calN}\beta_i h_i(x)\in \calC$ such that $x_1\in \argmin~ p_1(x)$ and $x_2\in \argmin~ p_2(x)$, respectively. Note that it is possible that $p_1(\cdot)=p_2(\cdot)$, and that $p_i(\cdot)=\widetilde{p}(\cdot)$ for $i=1$ or $i=2$.\\

Given $0\le \alpha\le 1$, let $x_{\alpha}=\alpha x_1+(1-\alpha) x_2$. We consider two cases:
\begin{itemize}
\item[(i)] ~$x_{\alpha}\in \argmin ~p_1(x)\cup \argmin ~p_2(x) \cup \argmin ~ \widetilde{p}(x)$, and
\item[(ii)] $x_{\alpha}\notin \argmin ~p_1(x)\cup \argmin ~p_2(x) \cup \argmin ~ \widetilde{p}(x)$
\end{itemize}

\paragraph{{\bf Case (i)}:~$x_{\alpha}\in \argmin ~p_1(x)\cup \argmin ~p_2(x) \cup \argmin ~ \widetilde{p}(x)$}

When $x_{\alpha}\in \argmin ~p_1(x)\cup \argmin ~p_2(x) \cup \argmin ~ \widetilde{p}(x)$, by definition of $Y$, we have
$$ x_{\alpha}\in \argmin ~p_1(x)\cup \argmin ~p_2(x) \cup \argmin ~ \widetilde{p}(x)\subseteq Y.$$
Thus, $x_{\alpha}\in Y$.

\paragraph{{\bf Case (ii)}:~$x_{\alpha}\notin \argmin ~p_1(x)\cup \argmin ~p_2(x) \cup \argmin ~ \widetilde{p}(x)$}

By symmetry, WLOG, assume that $x_1<x_2$. By definition of $x_{\alpha}$, it holds that $x_1<x_{\alpha}<x_2$. By assumption of case (ii), it must be that $x_{\alpha}> \max \pth{\argmin p_1(x)}$ and $x_{\alpha}< \min \pth{\argmin p_1(x)}$, which imply that $p_1^{\prime}(x_{\alpha})>0$ and $p_2^{\prime}(x_{\alpha})<0$.

There are two possibilities for $\widetilde{p}^{\prime}(x_{\alpha})$
(the gradient of $\widetilde{p}(x_{\alpha})$):
$\widetilde{p}^{\prime}(x_{\alpha})<0$ or
$\widetilde{p}^{\prime}(x_{\alpha})>0$.
Note that $\widetilde{p}^{\prime}(x_{\alpha})\neq 0$
because $x_\alpha\not\in\argmin ~ \widetilde{p}(x)$.

~

When $\widetilde{p}^{\prime}(x_{\alpha})<0$, there exists $0\le \zeta\le 1$ such that
$$
 \zeta ~p_1^{\prime}(x_{\alpha}) + (1-\zeta)~\widetilde{p}^{\prime}(x_{\alpha})=0.$$
By definition of $p_1(x)$ and $\widetilde{p}(x)$, we have
 \begin{align*}
 0=\zeta ~p_1^{\prime}(x_{\alpha}) + (1-\zeta)~\widetilde{p}^{\prime}(x_{\alpha})&=\zeta ~\pth{\sum_{i\in \calN}\alpha_ih_i^{\prime}(x_{\alpha})} + (1-\zeta)\pth{\frac{1}{|\calN|}\sum_{i\in \calN}h_i^{\prime}(x_{\alpha})}\\
 &=\sum_{i\in \calN}\pth{ \alpha_i\zeta +(1-\zeta)\frac{1}{|\calN|}}h_i^{\prime}(x_{\alpha}).
\end{align*}
Thus, $x_{\alpha}$ is an optimum of function
\begin{align}
\label{valid obj non-iter1}
\sum_{i\in \calN}\pth{ \alpha_i\zeta +(1-\zeta)\frac{1}{|\calN|}}h_i(x).
\end{align}
Let $\calI$ be the collection of indices defined by
$$\calI\triangleq \{~i: ~ i\in \calN, ~\text{and}~ \alpha_i\zeta +(1-\zeta)\frac{1}{|\calN|}~~\ge~~ \frac{1}{2(n-f)}~\}.$$
Next we show that $|\calI|\ge n-2f$. Let $\calI_1$ be the collection of indices defined by
$$\calI_1\triangleq \{~i: ~ i\in \calN, ~\text{and}~ \alpha_i\ge \frac{1}{2(n-f)}~\}.$$
Since $p_1(x)\in \calC$, then $|\calI_1|\ge n-2f$. In addition, since $n>3f$, $|\calN|< 2(n-f)$.\footnote{In fact, $n>2f$ suffices for this.} Then, for each $j\in \calI_1$, we have
$$\alpha_i\zeta +(1-\zeta)\frac{1}{|\calN|}\ge \zeta \frac{1}{2(n-f)}+(1-\zeta)\frac{1}{|\calN|}> \zeta \frac{1}{2(n-f)}+(1-\zeta)\frac{1}{2(n-f)}= \frac{1}{2(n-f)},$$
i.e., $j\in \calI$. Thus, $\calI_1\subseteq \calI$.

Since $|\calI_1|\ge n-2f$, we have $|\calI|\ge n-2f$. So function (\ref{valid obj non-iter1}) is a valid function. Thus, $x_{\alpha}\in Y$. \\

Similarly, we can show that the above result holds when $\widetilde{p}^{\prime}(x_{\alpha})>0$.

Therefore, set $Y$ is convex.

\eproof
\end{proof}

\begin{lemma}
\label{gradient valid non-iter1}
For each iteration $t$, there exists a valid function $p(x)\in \calC$ such that $g[t-1]=p^{\prime}(x)$.
\end{lemma}
\begin{proof}
Let $\calR^*[t-1]$ denote the set of nodes from whom the remaining $n-2f$ values were received in iteration $t$, and let us denote by $\calL[t-1]$ and $\calS[t-1]$ the set of nodes from whom the largest $f$ values and the smallest $f$ values were received in iteration $t$. Due to the fact that each value is transmitted using Byzantine broadcast, $\calR^*[t-1]$, $\calL[t-1]$ and $\calS[t-1]$ do not depend on $j$, for $j\in \calN$. Thus, $\calR^*[t-1]$, $\calL[t-1]$ and $\calS[t-1]$ are well-defined.\\

By definition of $g[t-1]$, we have
\begin{align}
\label{rewritten gradient non-iter1}
\nonumber
g[t-1]&=\frac{1}{n-2f} \sum_{j\in \calR^*[t-1]} g_j[t-1]\\
\nonumber
&=\frac{1}{n-2f} \pth{\sum_{j\in \calR^*[t-1]-\calF} g_j[t-1]+\sum_{j\in \calR^*[t-1]\cap\calF} g_j[t-1]}\\
&=\frac{1}{n-2f} \pth{\sum_{j\in \calR^*[t-1]-\calF} h_j^{\prime}(x[t-1])+\sum_{j\in \calR^*[t-1]\cap\calF} g_j[t-1]}~~\text{since}~ \forall i\in \calN, g_i[t-1]=h_i^{\prime}(x[t-1])
\end{align}

We consider two cases: (i) $\calR^*[t-1]\cap\calF=\O$ and (ii) $\calR^*[t-1]\cap\calF\not=\O$.

\paragraph{{\bf Case (i)}: $\calR^*[t-1]\cap\calF=\O$.}

 When $\calR^*[t-1]\cap\calF=\O$, it holds that $\calR^*[t-1]\subseteq \calN$. Define $p(x)=\frac{1}{n-2f}\sum_{j\in \calR^*[t-1]} h_j(x)$. We have
\begin{align}
\nonumber
p^{\prime}(x[t-1])=\frac{1}{n-2f} \sum_{j\in \calR^*[t-1]} h_j^{\prime}(x[t-1])=g[t-1].
\end{align}
In addition, in function $p(x)$, $n-2f$ component functions, corresponding to functions in $\calR^*[t-1]$ have weights $\frac{1}{n-2f}\ge \frac{1}{2(n-f)}$. Thus $p(x)\in \calC$. This proves Lemma \ref{gradient valid non-iter1} for $\calR^*[t-1]\cap\calF=\O$.\\


\paragraph{{\bf Case (ii)}: $\calR^*[t-1]\cap\calF\not=\O$.}

Let $|\calR^*[t-1]\cap\calF|=\theta$, let $\calL^*[t-1]\subseteq \calL[t-1]\cap
\calN$ and $\calS^*[t-1]\subseteq \calS[t-1]\cap
\calN$ such that $|\calL^*[t-1]|=|\calS^*[t-1]|=\theta$.
Since $|\calF|\le f$ and $|\calR^*[t-1]\cap\calF|=\theta$, it follows that $|\calL[t-1]\cap \calF|\le f-\theta$ and $|\calS[t-1]\cap \calF|\le f-\theta$.
We have
$$|\calL[t-1]\cap
\calN|=|\calL[t-1]|-|\calL[t-1]\cap \calF|\ge f-(f-\theta)=\theta,$$
and
$$|\calS[t-1]\cap
\calN|=|\calS[t-1]|-|\calS[t-1]\cap \calF|\ge f-(f-\theta)=\theta.$$
Thus $\calL^*[t-1]$ and $\calS^*[t-1]$ are well-defined.

By definition of $\calL^*[t-1]$, $\calS^*[t-1]$ and $\calR^*[t-1]\cap\calF$, it follows that for each $i\in \calL^*[t-1]$, $j\in \calS^*[t-1]$ and $k\in \calR^*[t-1]\cap\calF$
$$h_i^{\prime}(x[t-1])\ge g_k[t-1]\ge h_j^{\prime}(x[t-1]). $$
Since $|\calL^*[t-1]|=|\calS^*[t-1]|=|\calR^*[t-1]\cap \calF|=\theta$, we have
$$\sum_{i\in \calL^*[t-1]} h_i^{\prime}(x[t-1])\ge \sum_{k\in \calR^*[t-1]\cap\calF} g_k[t-1]\ge \sum_{j\in \calS^*[t-1]} h_j^{\prime}(x[t-1]).$$
Thus, there exists $0\le \zeta \le 1$ such that
$$\sum_{k\in \calR^*[t-1]\cap\calF} g_k[t-1]=\zeta \pth{\sum_{i\in \calL^*[t-1]} h_i^{\prime}(x[t-1])}+(1-\zeta)\pth{\sum_{j\in \calS^*[t-1]} h_j^{\prime}(x[t-1])}.$$
So (\ref{rewritten gradient non-iter1}) can be rewritten as
\begin{align*}
g[t-1]=\frac{1}{n-2f} \pth{\sum_{j\in \calR^*[t-1]-\calF} h_j^{\prime}(x[t-1])+\zeta \sum_{i\in \calL^*[t-1]} h_i^{\prime}(x[t-1])+(1-\zeta)\sum_{j\in \calS^*[t-1]} h_j^{\prime}(x[t-1])}.
\end{align*}
Define $q(x)$ as
$$q(x)=\frac{1}{n-2f} \pth{\sum_{j\in \calR^*[t-1]-\calF} h_j(x[t-1])+\zeta \sum_{i\in \calL^*[t-1]} h_i(x[t-1])+(1-\zeta)\sum_{j\in \calS^*[t-1]} h_j(x[t-1])}.$$
By definition of $q(x)$, it holds that $g[t-1]=q^{\prime}(x[t-1])$. Next we show $q(x)\in \calC$.

Since $|\calL^*[t-1]|=|\calS^*[t-1]|=|\calR^*[t-1]\cap \calF|=\theta$, it holds that
\begin{align*}
&\frac{1}{n-2f} \pth{\sum_{j\in \calR^*[t-1]-\calF} 1+\zeta \sum_{i\in \calL^*[t-1]} 1+(1-\zeta)\sum_{j\in \calS^*[t-1]} 1}\\
&\quad=\frac{1}{n-2f} \pth{\sum_{j\in \calR^*[t-1]-\calF} 1+\zeta \sum_{i\in \calR^*[t-1]\cap \calF} 1+(1-\zeta)\sum_{j\in \calR^*[t-1]\cap \calF} 1}\\
&\quad=\frac{1}{n-2f} \pth{\sum_{j\in \calR^*[t-1]-\calF} 1+\sum_{i\in \calR^*[t-1]\cap \calF} 1}=\frac{1}{n-2f} \sum_{j\in \calR^*[t-1]} 1=1.
\end{align*}
Thus function $q(x)$ is a convex combination of $h_i(x)$ for $i\in \calN$.

When $\zeta\ge \frac{1}{2}$, it holds that all component functions in $\pth{\calR^*[t-1]-\calF}\cup \calL^*[t-1]$ have weights bounded below by $\frac{1}{2(n-2f)}\ge \frac{1}{2(n-f)}$.
In addition,
\begin{align*}
|\pth{\calR^*[t-1]-\calF}\cup \calL^*[t-1]|&=|\calR^*[t-1]-\calF|+|\calL^*[t-1]|\\
&=|\calR^*[t-1]|-|\calR^*[t-1]\cap \calF|+|\calL^*[t-1]|\\
&=n-2f-\theta+\theta=n-2f.
\end{align*}
Thus, $q(x)\in \calC$. This completes the proof.
\eproof
\end{proof}

\begin{definition}
Given a point $x$ and a nonempty set $C$, the distance of point $x$ to set $C$, denoted by $Dist(x, C)$ is defined by
$$ Dist(x, C)\triangleq \inf_{y\in C} \norm{x-y}_2 . $$
\end{definition}

\begin{theorem}
\label{converge of alg4}
When $n>3f$, Algorithm 4 solves Problem 3 with $\beta=\frac{1}{2(n-f)}$ and $\gamma=n-2f$.
\end{theorem}
\begin{proof}

Let $\{x[t]\}_{t=1}^{\infty}$ be the sequence of estimates generated by Algorithm 4. To show this theorem, it is enough to show the limit of $Dist(x[t], Y)$ exists and
$$\lim_{t\diverge} Dist(x[t], Y)=0.$$

We say that an element $x[t]$ is a resilient point if conditions
in one of the following items hold true:
\begin{list}{}{}
\item
\hspace*{1in}*
$~ x[t-1]\in Y~\text{and}~~x[t]\notin Y,$\\
\hspace*{1in}*
$~ x[t-1]\ge \sup Y~\text{and}~~x[t]\le \inf Y,$\\
\hspace*{1in}*
$~ x[t-1]\le \inf Y~\text{and}~~x[t]\ge \sup Y.$
\end{list}

We consider two cases: (i) there are infinity many resilient points in $\{x[t]\}_{t=1}^{\infty}$, and (ii) there are finitely many resilient points in $\{x[t]\}_{t=1}^{\infty}$.

\paragraph{{\bf Case (i)}: There are infinity many resilient points in $\{x[t]\}_{t=1}^{\infty}$.}

Let $\{x[t_i]\}_{i=1}^{\infty}$ be the maximal subsequence of $\{x[t]\}_{t=1}^{\infty}$ such that each $x[t_i]$ is a resilient point.

Recall that in the update function (\ref{update valid non-iter1}),
$$x[t]=x[t-1]-\lambda[t-1]g[t-1]. $$
For each resilient point $x[t_i]$, it holds that
$$Dist(x[t_i], Y)\le |\lambda[t_i-1]g[t_i-1]|\le \lambda[t_i-1]L. $$
Thus,  $$\limsup_{i\diverge} ~ Dist(x[t_i], Y)\le \limsup_{i\diverge} ~ \lambda[t_i-1]L\overset{(a)}{=} \pth{\limsup_{i\diverge} ~ \lambda[t_i-1]}L =0. $$
Equality $(a)$ follows from the fact that the stepsize $\lambda[t]$ is diminishing, i.e., $\lim_{t\diverge} \lambda[t]=0$.

By definition,  for each $j$ such that $t_i< j< t_{i+1}$, the element $x[j]$ is not a resilient point. Thus, $x[j]\ge \sup~ Y$ for each $t_i< j< t_{i+1}$, or $x[j]\le \inf~ Y$ for each $t_i< j< t_{i+1}$ or $x[j]\in Y$ for each $t_i< j< t_{i+1}$. By the update function (\ref{update valid non-iter1}), we have
\begin{align*}
Dist(x[j], Y)&= ~\max \left\{0, ~Dist(x[j-1], Y)-\lambda[j-1]\, |g[j-1]|\right\}\\
&= ~\max\{0, ~Dist(x[j-1], Y)\}~~~\text{since}~~|g[j-1]|\ge 0\\
&\le ~Dist(x[j-1], Y)~~~\text{since}~~Dist(x[j-1], Y)\ge 0
\end{align*}
for each $t_i< j< t_{i+1}$ and each $i$, and all the above three  possible scenarios.
Consequently, we have
$$Dist(x[j], Y)\le Dist(x[t_i], Y).$$

Thus,
$$\limsup_{j\diverge} ~ Dist(x[j], Y)\le \limsup_{i\diverge} ~ Dist(x[t_i], Y)=0. $$
Since $Dist(x, Y)\ge 0$ for all $x$, we have $\liminf_{j\diverge} ~ Dist(x[j], Y)\ge 0$. Then,
$$\limsup_{j\diverge} ~ Dist(x[j], Y)\le 0\le \liminf_{j\diverge} ~ Dist(x[j], Y).$$
On the other hand, by definition of $\liminf$ and $\limsup$, we have
$$\liminf_{j\diverge} ~ Dist(x[j], Y)\le \limsup_{j\diverge} ~ Dist(x[j], Y).$$
Thus,
$$\liminf_{j\diverge} ~ Dist(x[j], Y)=\limsup_{j\diverge} ~ Dist(x[j], Y)=0.$$
Therefore, the limit of $Dist(x[j], Y)$ exists and
$$\lim_{j\diverge} Dist(x[j], Y)=0.$$\\

\paragraph{{\bf Case (ii)}: There are finitey many resilient points in $\{x[t]\}_{t=1}^{\infty}$. }
Let $t_0$ be the largest index such that $x[t_0]$ is a resilient point. If there exists $t^{\prime}\ge t_0$ such that $x[t^{\prime}]\in Y$, then $x[t]\in Y$ for all $t\ge t^{\prime}$, i.e., $Dist(x[t], Y)=0$, for all $t\ge t^{\prime}$. Thus, the limit of $Dist(x[t], Y)$ exists and
$$\lim_{t\diverge} Dist(x[t], Y)=0.$$

Now we consider the scenario where $x[t]\notin Y$ for all $t\ge t_0$. Since $t_0$ is the largest index such that $x[t_0]$ is a resilient point, then either $x[t]\le \inf Y$ for all $t\ge t_0$ or $x[t]\ge \sup Y$ for all $t\ge t_0$.

By symmetry, it is enough to consider the case when $x[t]\le \inf Y$ for all $t\ge t_0$. By update function (\ref{update valid non-iter1}), and the definition of $t_0$, it can be seen that the subsequence $\{x[t]\}_{t=t_0}^{\infty}$ is an increasing sequence. In addition, we know $x[t]\le \inf Y\in \reals$ for all $t\ge t_0$. By Monotone Convergence Theorem, the limit of $\{x[t]\}_{t=t_0}^{\infty}$ exists and $\lim_{t\diverge}x[t]=x^*\le \inf Y$.\\

If $\lim_{t\diverge}x[t]=x^*= \inf Y$, then $\lim_{t\diverge}Dist(x[t], Y)=0$.

~

Now we consider the case when $\lim_{t\diverge}x[t]=x^*<\inf Y$.  Since $x^*<\inf Y$, there exists $\epsilon >0$ such that $x^*=\inf Y-\epsilon$. Let $\rho=\sup_{p(\cdot)\in \calC} p^{\prime}(x^*)$. We show next that $\rho<0$
{\bf provided that}
$x^*=\inf Y-\epsilon$.

~

For each $p(\cdot)\in \calC$, $p^{\prime}(x^*)<0$. Then, $\rho=\sup_{p(\cdot)\in \calC} p^{\prime}(x^*)\le 0$.

Let $h_{i_1}^{\prime}(x^*), \cdots, h_{i_{|\calN|}}^{\prime}(x^*)$ be a non-increasing order of $h_j^{\prime}(x^*)$, for $j\in \calN$. Define $q(x)$ as follows,
$$q(x)=\pth{1-\frac{n-2f-1}{2(n-f)}} h_{i_1}(x)+\frac{1}{2(n-f)}\sum_{j=2}^{n-2f} h_{ij}(x).$$

It can be easily seen that $q(\cdot)\in \calC$ is a valid function and $$\sup_{p(\cdot)\in \calC} p^{\prime}(x^*)=q^{\prime}(x^*).$$

Thus, if $x^*=\inf~Y-\epsilon$ for some $\epsilon>0$, then $\rho=q^{\prime}(x^*)<0$.

~

\noindent
From the update function (\ref{update valid non-iter1}), we have
\begin{align*}
x[t+m+1]&=x[t+m]-\lambda[t+m]g[t+m]\\
&=x[t_0]-\sum_{j=t_0}^{t+m}\lambda[j]g[j]\\
&\overset{(a)}{\ge} x[t_0]-\sum_{j=t_0}^{t+m}\lambda[j]\rho.
\end{align*}
Inequality $(a)$ is true because (1) the gradient of a convex function is non-decreasing, (2) $x[t]\le x^*$ for each $t\ge t_0$, and (3) $\rho=\sup_{p(\cdot)\in \calC} p^{\prime}(x^*)$.

 Let $m\diverge$, we have

\begin{align*}
\lim_{m\diverge}  x[t+m+1]&~\ge~  x[t_0]-\lim_{m\diverge}\sum_{j=t_0}^{t+m}\lambda[j]\rho\\
&~=~x[t_0]-\pth{\lim_{m\diverge}\sum_{j=t_0}^{t+m}\lambda[j]}\rho\\
&~=~x[t_0]+\infty=\infty.
\end{align*}
On the other hand, we know $\lim_{m\diverge}  x[t+m+1]\le x^*\in \reals$. A contradiction is proved.
Thus, $x^*=\inf Y$. Consequently, $\lim_{t\diverge}Dist(x[t], Y)=0$. \\

This completes the proof.

\eproof
\end{proof}
\subsection{Algorithm 5: Non-Interleaved Iterative Algorithm II}
As commented before, the performance guarantee of Algorithm 4, in contrast to Algorithms 1, 2, and 3, is independently of $|\calN|$.  Algorithm 5, described below, admits the analysis between the tradeoff of $\beta$ and $\gamma$ in terms of $|\calN|$.

\paragraph{}
\hrule
~
\vspace*{4pt}

{\bf Algorithm 5} for agent $j$:
~
\vspace*{4pt}\hrule

\begin{list}{}{}
\item[{\bf Initialization Step (i):}]
Choose $x_j[0]\in X_j=\arg\min_{x\in\mathbb{R}} h_j(x)$.\\
~
\item[{\bf Initialization Step (ii):}] Perform exact Byzantine consensus with $x_j[0]$ as the input of agent $j$. Set $x[0]$ to the output of the above consensus algorithm.\\

\item[{\bf Iteration $t\geq 1$:}]

\begin{itemize}
\item[{\bf Step 1:}] Compute $h_j^{\prime}\pth{x_j[t-1]}$, and perform Byzantine broadcast of $h_j^{\prime}\pth{x_j[t-1]}$ to all the agents, using any Byzantine broadcast algorithm, such as \cite{psl_BG_1982}.\\

    In step 1, agent $j$ should receive a gradient from each agent $i\in\calV$ -- let us denote the gradient received from agent $i$ in iteration $t$ as $g_i[t-1]$.\footnote{For each $i\in \calN$, $g_i[t-1]=h_i^{\prime}(x[t-1])$.} If no gradient is, in fact, received from agent $i$ in iteration $t$
via a Byzantine broadcast from $i$, then it must be the case that agent $i$ is faulty. In that case,
 agent $i$ is isolated (i.e., removed from the system).\footnote{Alternatively, the gradient value can be replaced by some default value (say 0).}
 This reduces the total number of agents $n$ by 1,
and the maximum number of faulty agent $f$ is also reduced by 1.
Algorithm 5 is {\em restarted} (from Step 1) using the new parameters
$n$ and $f$.\\

\item[{\bf Setp 2:}] Due to the restart mechanism above,
the algorithm progresses to Step 2 only when agent $j$ has received gradients from all other nodes.
Let $\calR[t-1]$ denote the multiset of the gradients $\{g_1[t-1],g_2[t-1],\cdots,g_n[t-1]\}$, obtained in Step 1 of $t$-th iteration. Drop the smallest $f$ values and the largest $f$.
Denote the largest and smallest gradients among the remaining values by $\hat{g}[t-1]$ and $\check{g}[t-1]$, respectively. Set $g[t-1]=\frac{1}{2}\pth{\hat{g}[t-1]+\check{g}[t-1]}$.

\item[{\bf Setp 3:}] Let $\{\lambda[t]\}_{t=0}^{\infty}$ be a sequence of diminishing (non-increasing and $\lambda[t]\to 0$) stepsizes chosen beforehand such that $\lambda[t]>0$ for each $t$, $\sum_{t=1}^{\infty}\lambda[t]=\infty$ and $\sum_{t=1}^{\infty}\lambda^2[t]<\infty$.

Compute $x[t]$ as
\begin{align}
\label{update valid}
x[t]= x[t-1]-\lambda[t-1] g[t-1].
\end{align}

\end{itemize}

\end{list}

~
\hrule

~

~

Since $x[0]$ is the consensus value of the input $x_j[0]$ for each $j\in \calN$, and $g[0]$ is the consensus value of the non-faulty gradient $h_j^{\prime}(x[0])$, for each $j\in \calN$, the update function (\ref{update valid}) is well-defined for $t=1$. By an inductive argument, we can show that the update (\ref{update valid}) is well-defined for all $t$.\\

Let $\widetilde{\calC}$ be the collection of functions defined as follows:
\begin{align}
\nonumber
\widetilde{\calC}\triangleq \{~~~p(x): p(x)&=\sum_{i\in \calN} \alpha_i h_i(x), ~~\forall i\in\calN, ~ \alpha_i\geq 0,\\
\nonumber
&\sum_{i\in \calN}\alpha_i=1,\text{~~and~~}\\
&\sum_{i\in\calN} {\bf 1}\left(\alpha_i\ge \frac{1}{2(|\calN|-f)}\right) ~ \geq ~ |\calN|-f ~~~ \}
\label{valid collection}
\end{align}

Each $p(x)\in \widetilde{\calC}$ is called a valid function. Note that the function $\frac{1}{|\calN|}\sum_{i\in \calN}h_i(x)\in \widetilde{\calC}$. For ease of future reference, we let $\widetilde{p}(x)=\frac{1}{|\calN|}\sum_{i\in \calN}h_i(x)$. Define $\widetilde{Y}\triangleq \cup_{p(x)\in \widetilde{\calC}} \argmin~ p(x)$.

\begin{lemma}
\label{valid hull non-iter2}
$\widetilde{Y}$ is a convex set.
\end{lemma}
The proof of Lemma \ref{valid hull non-iter2} is similar to the proof of Lemma \ref{valid hull non-iter1}, and is presented in Appendix \ref{app: valid hull non-iter2}.

\begin{lemma}
\label{gradient valid non-iter2}
For each iteration $t$, there exists a valid function $p(x)\in \widetilde{\calC}$ such that $g[t-1]=p^{\prime}(x[t-1])$.
\end{lemma}
The proof of Lemma \ref{gradient valid non-iter2} is presented in Appendix \ref{app: gradient valid non-iter2}.

\begin{theorem}
\label{converge of alg5}
When $n>3f$, Algorithm 5 solves Problem 3 with $\beta=\frac{1}{2(|\calN|-f)}$ and $\gamma=|\calN|-f$.
\end{theorem}
The proof of Theorem \ref{converge of alg5} is similar to the proof of Theorem \ref{converge of alg4}, and is omitted. 

\section{Algorithm 6: Suboptimal Algorithm}\label{approximation alg}
Algorithms 1, 2, 3, 4 and 5 all use Byzantine broadcast as subroutines, which may be costly. 
Unlike these algorithms, the iterative Algorithm 6 presented below \emph{does not} require the agents to exchange their local cost functions and gradients. Instead, in Algorithm 6, each agent optimizes its local cost function locally and exchanges the local optima, using an arbitrary Byzantine consensus algorithm. In addition, the correctness proof of Algorithm 6 does not require each $h_i(\cdot)$ to be differentiable. Thus Algorithm 6 also works for non-smooth functions.

Algorithm 6 is not an optimal algorithm. Specifically, Algorithm 6 only solves Problem 3 with $\beta=\frac{1}{2|\calN|}$ and $\gamma=\lceil\frac{n}{2}\rceil-\phi$, instead of the optimal $\gamma^*=|\calN|-f$ achieved by Algorithms 1, 2, 3 and 5. That is, Algorithm 6 is a suboptimal Algorithm for Problem 3 with $\beta=\frac{1}{2|\calN|}$.

\paragraph{}
\vspace*{8pt}\hrule
~

{\bf Algorithm 6} for agent $j$:
~
\vspace*{4pt}\hrule

\begin{list}{}{}
\item[{\bf Step 1:}]
Choose $v_j\in X_j=\argmin_{x\in\mathbb{R}} h_j(x)$.\\
~
\item[{\bf Step 2:}]
Send $v_j$ to all agents, and receive messages from all agents. Agent $j$ should receive a value from each agent $i\in \calV$--let us denote the value received from agent $i$ as $w_{ij}$. If no value is, in fact, received from agent $i$, then $w_{ij}$ is set to be a predefined default value.

Sort $w_{ij}$ in a non-decreasing order, breaking tie arbitrarily,
and set $x_j[0]$ to be the median of this order, i.e., we choose $x_j[0]$ to be the $w_{ij}$ whose rank is $\lceil \frac{n}{2}\rceil$.\\
~

\item[{\bf Step 3:}] Perform exact Byzantine consensus algorithm with $x_j[0]$ as
the input of agent $j$ to the consensus algorithm.\\

Set $\tx$ to be the output of the above consensus algorithm, and output $\tx$.

\end{list}

\hrule

~

\begin{theorem}
\label{lbnontrivial}
When $n>3f$, Algorithm 6 solves Problem 3 with $\beta=\frac{1}{2|\calN|}$ and $\gamma=\lceil \frac{n}{2}\rceil-\phi$.
\end{theorem}
\begin{proof}

Let $W_j$ denote the multiset obtained by agent $j$, i.e., $W_j=\{w_{1j}, \ldots, w_{nj}\}$. For each $x\in \reals$, define $W_j^+(x)$ and $W_j^-(x)$ as follows.
$$
 W_j^+(x)=\{i: ~i\in \calN~\text{and}~ w_{ij}\ge x\},
$$
$$
W_j^-(x)=\{i: ~i\in \calN~\text{and}~ w_{ij}\le x\}.
$$ 
Note that $W_j^+(x)\cup W_j^-(x)=\calV$ for each $x\in \reals$, and that $w_{ij}=v_i$ for each $i\in \calN$. It should also be noted that $W_j^+(x)$ and $W_j^-(x)$ are not necessarily disjoint.\\

For each $j$, since $x_j[0]$ is chosen to be the median of the non-decreasing order over $W_j$, we have

$$
 |W_j^+(x_j[0])|=\left |\{i: ~i\in \calN~\text{and}~ w_{ij}\ge x_j[0]\}\right |\ge \lceil \frac{n}{2}\rceil-\phi,
$$
$$
 |W_j^-(x_j[0])|=\left |\{i: ~i\in \calN~\text{and}~ w_{ij}\le x_j[0]\}\right |\ge n-\lceil \frac{n}{2}\rceil-\phi+1\ge \lceil \frac{n}{2}\rceil-\phi.
$$
Let $i_0\in \calN$ and $j_0\in \calN$ be the agents such that $x_{i_0}[0]\le x_j[0]$ for each $j\in \calN$ and $x_{j_0}[0]\ge x_j[0]$ for each $j\in \calN$. Since $\tx$ is the output of a correct exact consensus algorithm, by validity, we have $x_{i_0}[0]\le \tx \le x_{j_0}[0]$. Thus
$$ \{i: i\in \calN, w_{ii_0}\le x_{i_0}[0]\}\subseteq \{i: i\in \calN, w_{ii_0}\le \tx\}=\{i: i\in \calN, v_{i}\le \tx\}$$
and
$$ \{i: i\in \calN, w_{ij_0}\ge x_{j_0}[0]\}\subseteq \{i: i\in \calN, w_{ij_0}\ge \tx\}=\{i: i\in \calN, v_{i}\ge \tx\}.$$
Consequently, we have
\begin{align}
\label{smaller}
|\{i: i\in \calN, v_{i}\le \tx\}|\ge |\{i: i\in \calN, w_{ii_0}\le x_{i_0}[0]\}|=~|W_{i_0}^-(x_{i_0}[0])|~\ge ~ \lceil\frac{n}{2}\rceil-\phi,
\end{align}
and
\begin{align}
\label{larger}
|\{i: i\in \calN, v_{i}\ge \tx\}|\ge |\{i: i\in \calN, w_{ij_0}\ge x_{j_0}[0]\}|=~|W_{j_0}^+(x_{j_0}[0])|~\ge ~ \lceil\frac{n}{2}\rceil-\phi.
\end{align}

Recall that $v_j\in X_j=\argmin_{x\in\mathbb{R}} h_j(x)$, then $h_i(\tx)\ge 0$ for each $i\in \{i: i\in \calN, v_{i}\le \tx\}$, and $h_i(\tx)\le 0$ for each $i\in \{i: i\in \calN, v_{i}\ge \tx\}$. Define $A(\tx), B(\tx)$ and $C(\tx)$ as follows.
\begin{align*}
A(\tx)&\triangleq \{i: i\in \calN, ~ h_i^{\prime}(\tx)>0\},\\
B(\tx)&\triangleq \{i: i\in \calN, ~ h_i^{\prime}(\tx)<0\},\\
C(\tx)&\triangleq \{i: i\in \calN, ~ h_i^{\prime}(\tx)=0\}.
\end{align*}

We now consider two cases: (i) $A(\tx)=\O$ or $B(\tx)=\O$, and (ii) $A(\tx)\not=\O$ and $B(\tx)\not=\O$.

\paragraph{{\bf Case (i)}: $A(\tx)=\O$ or $B(\tx)=\O$.}

If $B(\tx)=\O$, then $h_i(\tx)=0$ for each $i\in \{i: i\in \calN, v_{i}\ge \tx\}$. Then $\tx$ is an optimum of function
\begin{align}
\label{rewritten3}
\frac{1}{\left | \{i: ~i\in \calN,~ v_i\le \tx\} \right |}\sum_{j\in \{i: ~i\in \calN,~ v_i\le \tx\} } h_j(x).
\end{align}
As $\left | \{i: ~i\in \calN,~ v_i\le \tx\} \right |\le |\calN|$ and by (\ref{smaller}), it holds that $\left | \{i: ~i\in \calN,~ v_i\le \tx\} \right |\ge \lceil\frac{n}{2}\rceil-\phi$. Thus, in (\ref{rewritten3}) at least $\lceil \frac{n}{2}\rceil-\phi$ non-faulty functions are assigned coefficients bounded below by $\frac{1}{|\calN|}$.\\

Similarly, we can show the case when $A(\tx)=\O$.

\paragraph{{\bf Case (ii)}: $A(\tx)\not=\O$ and $B(\tx)\not=\O$.}
When $A(\tx)\not=\O$ and $B(\tx)\not=\O$,
$$
\sum_{i\in A(\tx) }h_i^{\prime}(\tx)~>~0 ~~~~~~\text{and}~~~~~~\sum_{i\in B(\tx) }h_i^{\prime}(\tx)~<~0.
$$
Then there exists $0\le \zeta\le 1$ such that
\begin{align*}
0=\zeta \pth{\sum_{i\in A(\tx) }h_i^{\prime}(\tx)}+\pth{1-\zeta} \pth{\sum_{i\in B(\tx) }h_i^{\prime}(\tx)}.
\end{align*}
In addition, by definition of $C(\tx)$, we have
\begin{align*}
\zeta \pth{\sum_{i\in A(\tx) }h_i^{\prime}(\tx)}+\pth{1-\zeta} \pth{\sum_{i\in B(\tx) }h_i^{\prime}(\tx)} +\sum_{i\in C(\tx)}h_i^{\prime}(\tx)&=0+\sum_{i\in C(\tx)}h_i^{\prime}(\tx)\\
&=0+0=0.
\end{align*}
Thus $\tx$ is an optimum of
\begin{align}
\label{rewritten4}
 \chi\pth{\zeta \sum_{i\in A(\tx)} h_i(x)+\pth{1-\zeta} \sum_{i\in B(\tx)} h_i(x)+\sum_{i\in C(\tx)} h_i(x)},
\end{align}
where $$\chi=\frac{1}{\zeta |A(\tx)|+\pth{1-\zeta} |B(\tx)|+|C(\tx)|}.$$
Since $0\le \zeta \le 1$, either $\zeta\ge \frac{1}{2}$ or $1-\zeta \ge \frac{1}{2}$. WLOG, assume $\zeta\ge \frac{1}{2}$. We have
$$\zeta |A(\tx)|+\pth{1-\gamma} |B(\tx)|+|C(\tx)|\le |A(\tx)|+|B(\tx)|+|C(\tx)|=|A(\tx)\cup B(\tx)\cup C(\tx)|=|\calN|.$$
In addition, since $A(\tx)\cup C(\tx)
\supseteq \{i: i\in \calN~\text{and}~v_i\le \tx\}$ and $B(\tx)\cup C(\tx)\supseteq\{i: i\in \calN~\text{and}~v_i\ge \tx\}$, by definition of $\tx$, we have
$|A(\tx)\cup C(\tx)|\ge \lceil \frac{n}{2}\rceil-\phi$ and $|B(\tx)\cup C(\tx)|\ge \lceil \frac{n}{2}\rceil-\phi$.
 Then in (\ref{rewritten4}), at least $\lceil \frac{n}{2}\rceil-\phi$ non-faulty functions are assigned with weights at least $\frac{1}{2|\calN|}$. Similar result holds when $1-\zeta\ge \frac{1}{2}$.\\

Cases (i) and (ii) together prove the theorem.
\eproof
\end{proof}

\section{Extensions}\label{extension}

Many extensions of these results are possible. The results obtained in this technical report can be extended to the case when the functions
are sub-differentiable, with slightly more involved analysis.
This generalization will be presented in another technical report.

We have also obtained a comparable set of results
for the case when the cost functions are {\em redundant} in some manner
(e.g., cost function of agent 3 may equal a convex combination of cost functions
of agents 1 and 2), or the optimal sets of the local cost functions
are guaranteed to overlap. These results will also be presented elsewhere.

Finally, if the underlying communication channel is a broadcast channel (over which all transmissions are received correctly and identically by all agents), then the results presented in this report can be proved for $n\ge 2f+1$.

\section{Summary}\label{conclusion and discussion}

In this paper, we introduce the problem of Byzantine fault-tolerant optimization,
and obtain an impossibility result for the problem.
The impossibility result provides an upper bound on the number of
local cost function of non-faulty nodes that can {\em non-trivially}
affect the output of a weighted optimization function.
We also present algorithms that matches this upper bound. In addition, a low-complexity suboptimal algorithm is presented.

%
%

\newpage

\bibliographystyle{plain}
\bibliography{PSDA_DL}

\newpage
\appendix


\input{appendix}

\end{document}

%% file: appendix.tex
\centerline{\large\bf Appendices}

~

\section{Proposition \ref{p1}}
\label{appendix:p1}

Proposition \ref{p1} is used in proving the correctness of other results in the paper.

\begin{proposition}
\label{p1}
Let $\alpha_i\ge 0$ for $i\in\calN$ and $\sum_{i\in \calN}\alpha_i=1$. Consider
admissible functions $h_i(x)$, $i\in \calN$, with $X_i = \arg\min_{x\in\mathbb{R}}h_i(x)$.
Define $X$ as
\begin{align}
\label{convexcombination}
X = \arg\min_{x} \sum_{i\in \calN} \alpha_i h_i(x).
\end{align}
Then
\begin{align}
X\subseteq Cov\pth{\cup_{i\in \calN} X_i},
\end{align}
where $Cov\pth{\cup_{i\in \calN} X_i}$ is the convex hull of set $\cup_{i\in \calN} X_i$.
\end{proposition}
\begin{proof}
By definition of admissible functions, $X_i$ is nonempty for all $i\in \calN$. Then $Cov\pth{\cup_{i\in \calN} X_i}\not=\O$.

If $X=\O$, then $X\subseteq Cov\pth{\cup_{i\in \calN} X_i}$ holds trivially.

 It remains to be shown that when $X\not=\O$, $X\subseteq Cov\pth{\cup_{i\in \calN} X_i}$ is also true. We prove this by contradiction.

Suppose that $X\not=\O$ and $X\not\subseteq Cov\pth{\cup_{i\in \calN} X_i}$. Then there exists a value $x_0\in X$ that is not contained in $Cov\pth{\cup_{i\in \calN} X_i}$. Recall that, by definition
of admissible functions, $X_i$ is compact (closed and bounded) for each $i\in \calN$. Then $Cov\pth{\cup_{i\in \calN} X_i}$ is both convex and compact. To simplify notation, let $[a, b]\triangleq Cov\pth{\cup_{i\in \calN} X_i}$. In addition, $\sum_{i\in \calN}\alpha_i h_i^{\prime}(x_0)$ is the gradient of the function $\sum_{i\in \calN}\alpha_i h_i(x)$ at $x=x_0$.

As $x_0\not\in Cov\pth{\cup_{i\in \calN} X_i}$, then either $x_0<a$ or $x_0>b$.
By definition, $a$ is the smallest point at which there exists $i\in \calN$ such that $h_i^{\prime}\pth{x_0}=0$. If $x_0<a$, then $h_i^{\prime}\pth{x_0}<0$ for each $i\in \calN$. Otherwise the minimality of $a$ will be violated. In addition, since $\alpha_i\ge 0$ for each $i\in \calN$ and $\sum_{i\in \calN}\alpha_i=1$, then $\sum_{i\in \calN}\alpha_i h_i^{\prime}(x_0)<0$. However, since $x_0\in X$, by optimality of $x_0$ it must be that $\sum_{i\in \calN}\alpha_i h_i^{\prime}(x_0)=0$. This leads to a contradiction.

Similarly, we can derive a contradiction when $x_0>b$.
\eproof
\end{proof}

\section{Proof of Theorem \ref{t_imposs_0}}
\label{appendix:imposs:0}

\begin{proof}

Assume that $f>0$.

The proof of the theorem is by contradiction.

Suppose that there exists a correct algorithm $\calA$  that solves Problem 1.
Define the cost functions of the $n$ agents as follows.
\begin{itemize}
\item $h_1(x)=(x+1)^2$,
\item $h_{n}(x)=(x-1)^2$, and
\item $h_{i}(x)=x^2+i$, where $2\leq i\leq n-1$.
\end{itemize}
Let $X_i$ be the optimal set of $h_i(x)$ for $i\in\calV$.
That is, $X_i=\arg\max_{x\in\mathbb{R}}h_i(x)$. It is easy to see that $X_1=\{-1\}, X_{n}=\{1\}$, and for $2\leq i\leq n-1$, $X_i=\{0\}$. We consider two executions wherein $\calA$ produces different outputs,
and show that there exists a non-faulty agent that cannot distinguish these two executions.

The identity of the faulty agents in these two executions are different.
In both executions, the faulty nodes follow algorithm correctly with the above choice of cost functions.

\textbf{Execution 1}:
In execution 1, let $\calN=\{1, \cdots, n-1\}$ and $\calF=\{n\}$.
Since $\calA$ is a correct algorithm, by Proposition \ref{p1} it follows that
the output of the algorithm must be in $Cov\pth{\cup_{j=1}^{n-1} X_j}=[-1,0]$ for all agents $i\in \{1,\cdots,n-1\}$
-- note that Proposition \ref{p1} is stated and proved in Appendix \ref{appendix:p1}.

\textbf{Execution 2}:
In execution 2, let $\calN=\{2, \cdots, n\}$ and $\calF=\{1\}$.
Since $\calA$ is a correct algorithm, by Proposition \ref{p1} it follows that,
in this case, the output of the algorithm must be in $Cov\pth{\cup_{j=2}^n X_j}=[0,1]$ for all agents $i\in \{2,\cdots,n\}$

The agents in $\{2,\cdots,n-1\}$ cannot distinguish between the above two executions, and hence must produce identical output
in both cases. That is, their output must be 0 since $[-1,0] \cap [0,1]=\{0\}$.
(When $f>0$, $n\geq 3f+1=4$. Thus, the set $\{2,\cdots,n-1\}$ is non-empty.)

On the other hand, it is easy to see that $\sum_{i=1}^{n-1} h_i^{\prime}(0)\not=0$ and $\sum_{i=2}^{n} h_i^{\prime}(0)\not=0$, contradicting the hypothesis that $0$ is an optimal solution
for either execution--note that $h_i^{\prime}(x)$ is the derivative of function $h_i(\cdot)$ at $x$ for each $1\le i\le n$.
This contradicts the assumption that $\calA$ is correct and the proof is complete.
\eproof
\end{proof}

\section{Proof of Theorem \ref{ub1}}
\label{appendix:ub1}

\begin{proof}

Recall that we assume $n\ge 3f+1$ and that we denote $|\calF|=\phi$.
Let $h_1(x), \ldots, h_n(x)$ be defined as follows.
\begin{itemize}
\item $h_i(x)=\pth{x-i}^2$, for $1\le i\le f$ and $n-\phi+1\le i\le n$.\\
In this case, the optimum for $h_i(x)$ is at $x=i$.\\
\item $h_i(x)=\pth{x-a}^2$, for $f+1\le i\le n-\phi$, where $a=f+1$.\\
In this case, the optimum for $h_i(x)$ is at $x=a=f+1$.
\end{itemize}
From a non-faulty agent $j$'s perspective, any subset of $f$ agents may be faulty. Assume that, if agent $k$ is faulty, then aside from choosing its cost function
as specified above, agent $k$ does {\em not} behave incorrectly.
Thus, all agents follow any specified algorithm correctly.

To show the impossibility claim of the theorem, consider any correct algorithm.

Now, let us consider any non-faulty agent $j$ where $f+1\leq j\leq n-\phi$.
Consider two possible cases:
\begin{itemize}
\item[\bf Case 1:] In this case, suppose that agents 1 through $n-\phi$ are non-faulty, and agents $n-\phi+1$ through $n$ are faulty.  For the local cost functions (specified above) for the non-faulty agents in this case, the optima are in the interval $[1,a]$.  Then by Proposition \ref{p1}, for the output $\tx$ it must be true that $\tx\in [1, a]$.
(Recall that Proposition \ref{p1} is stated and proved in Appendix \ref{appendix:p1}.) \\

\item[\bf Case 2:] In Case 2, suppose that agents $f+1$ through $n$ are non-faulty, and agents $1$ through $f$ are faulty.  For the local cost functions (specified above) for the non-faulty agents in this case, the optima are in the interval $[a,n]$. Then by Proposition \ref{p1}, for the output $\tx$ it must be true that $\tx\in [a,n]$.
\end{itemize}
Since the non-faulty agent $j$ does not know the actual number of faulty agents in the system, it cannot distinguish between the above two cases, it must choose identical output in both cases.
Therefore, the output must be in $[1,a]\cap[a,n]$; that is, the output at non-faulty agent $j$ must equal $a=f+1$. Therefore,
all non-faulty agents must output $a=f+1$ in both cases.

Now suppose that Case 1 holds, i.e., agents $n-\phi+1$ through $n$ are faulty.
By the requirements of Problem 3,
there exists a collection of weights $\alpha_i$'s such that $\tx=a$ is an optimum of objective
\begin{align}
\sum_{i=1}^{n-\phi} \alpha_i h_i\pth{x},
\end{align}
Thus,
$\sum_{i=1}^{n-\phi} \alpha_i\, h_i^{\prime}(a)=0$, where $h_i^{\prime}(x)$ denotes the derivative of function $h_i(\cdot)$ at $x$.

Recall that $a=f+1$.
By construction of $h_1(x), \ldots, h_{n-\phi}(x)$, we know $h_i^{\prime}\pth{a}=0$ for $f+1\le i\le n-\phi$ and $h_i^{\prime}\pth{a}>0$ for $1\le i\le f$.
Thus
\begin{align*}
0=\sum_{i=1}^{n-\phi} \alpha_i h_i^{\prime}(a)
&=\pth{\sum_{i=1}^{f} \alpha_i h_i^{\prime}(a)}+\pth{\sum_{i=f+1}^{n-\phi} \alpha_i h_i^{\prime}(a)}\\
&=\pth{\sum_{i=1}^{f} \alpha_i h_i^{\prime}(a)}+\pth{\sum_{i=f+1}^{n-\phi} \alpha_i \times 0}\\
&=\sum_{i=1}^{f} \alpha_i h_i^{\prime}(a).
\end{align*}
For $1\le i\le f$, since $h_i^{\prime}(a)>0$ and $\alpha_i\ge 0$
it holds that $\alpha_ih_i^{\prime}(a)\ge 0$, where equality holds if and only if $\alpha_i=0$. Thus, $\sum_{i=1}^{f} \alpha_i h_i^{\prime}(a)=0$ implies that $\alpha_i h_i^{\prime}(a)=0$ for $1\le i\le f$. Then $\alpha_i=0$ for $1\le i\le f$.

Since there are $|\calN|$ non-faulty agents (1 through $n-\phi$), and weight $\alpha_i=0$ for $1\leq i\leq f$,
it follows that at most $|\calN|-f$ of the weights of the non-faulty agents in Case 1 are non-zero.

Thus, regardless of the value of parameter $\beta$ in Problem 3 (where $\beta>0$), if $\gamma$ exceeds $|\calN|-f$, no algorithm can solve Problem 3.
\eproof
\end{proof}

\section{Proof of Proposition \ref{monotone}}
\label{appendix:monotone}

\begin{proof}
We first show that $F\pth{x}$ is a non-decreasing function. 

Choose any $x\in\mathbb{R}$, and
choose any $y\ge x$. Let $S_y$ and $S_x$ be sets such that $\sum_{i\in A\pth{y}-S_y}h_i^{\prime}\pth{y}$ and $\sum_{i\in A\pth{x}-S_x}h_i^{\prime}\pth{x}$ are minimized, respectively.

Since $h_i\pth{\cdot}$ is convex, $h_i^{\prime}\pth{\cdot}$ is non-decreasing. By definition of $A\pth{\cdot}$ we have $A\pth{x}\subseteq A\pth{y}$, i.e., $A\pth{\cdot}$ is non-decreasing. In addition, $0\le |A\pth{\cdot}|\le n$. Similarly, we can show that $B(y)\subseteq B(x)$ and $0\le |B\pth{\cdot}|\le n$.
\begin{align}
\nonumber
F(y)-F(x)&=\sum_{i\in A\pth{y}-S_y}h_i^{\prime}\pth{y}-\sum_{i\in A\pth{x}-S_x}h_i^{\prime}\pth{x}\\
\nonumber
&=\sum_{i\in A\pth{y}-S_y-S_x}h_i^{\prime}\pth{y}+\sum_{i\in S_x\cap A\pth{y}-S_y}h_i^{\prime}\pth{y}-\pth{\sum_{i\in A\pth{x}-S_x-S_y}h_i^{\prime}\pth{x}+\sum_{i\in S_y\cap A\pth{x} -S_x}h_i^{\prime}\pth{x}}\\
\nonumber
&=\pth{\sum_{i\in A\pth{y}-S_y-S_x}h_i^{\prime}\pth{y}-\sum_{i\in A\pth{x}-S_x-S_y}h_i^{\prime}\pth{x}}+\pth{\sum_{i\in S_x\cap A\pth{y}-S_y}h_i^{\prime}\pth{y}-\sum_{i\in S_y\cap A\pth{x}-S_x}h_i^{\prime}\pth{x}}\\
\nonumber
&\overset{(a)}{\ge} \pth{\sum_{i\in A\pth{x}-S_y-S_x}h_i^{\prime}\pth{y}-\sum_{i\in A\pth{x}-S_x-S_y}h_i^{\prime}\pth{x}} +\pth{\sum_{i\in S_x\cap A\pth{y}-S_y}h_i^{\prime}\pth{y}-\sum_{i\in S_y\cap A\pth{x}-S_x}h_i^{\prime}\pth{x}}\\
\nonumber
&\overset{(b)}{=}\pth{\sum_{i\in A\pth{x}-S_y-S_x}h_i^{\prime}\pth{y}-\sum_{i\in A\pth{x}-S_x-S_y}h_i^{\prime}\pth{x}} +\pth{\sum_{i\in S_x-S_y}h_i^{\prime}\pth{y}-\sum_{i\in S_y\cap A\pth{x}-S_x}h_i^{\prime}\pth{x}}\\
\label{Fcases}
&\overset{(c)}{\ge} \pth{\sum_{i\in S_x-S_y}h_i^{\prime}\pth{y}-\sum_{i\in S_y\cap A\pth{x}-S_x}h_i^{\prime}\pth{x}}.
\end{align}
Inequality $(a)$ follows from the fact that $A(x)\subseteq A(y)$ and $h_i^{\prime}(y)>0$ for each $i\in A(y)$; equality $(b)$ is true since $S_x\subseteq A\pth{x}\subseteq A\pth{y}$; and inequality $(c)$ holds because that $h_i^{\prime}\pth{\cdot}$ is non-decreasing.

Now consider two cases: (i) $|S_x|<f$ and (ii) $|S_x|=f$.
\paragraph{\bf Case (i): $|S_x|<f$.}
In this case, we have $S_x=A(x)$, and
\begin{align}
\nonumber
\sum_{i\in S_x-S_y}h_i^{\prime}\pth{y}-\sum_{i\in S_y\cap A\pth{x}-S_x}h_i^{\prime}\pth{x}&= \sum_{i\in S_x-S_y}h_i^{\prime}\pth{y}-\sum_{i\in \O}h_i^{\prime}\pth{x}\\
\nonumber
&= \sum_{i\in S_x-S_y}h_i^{\prime}\pth{y}-0\\
\label{FcasesI}
&\ge 0.
\end{align}
\paragraph{\bf Case (ii): $|S_x|=f$.}
Because $S_x\subseteq A\pth{x}\subseteq A\pth{y}$, if $|S_x|= f$, we have $|A\pth{y}|\ge f$. Then, by definition of $S_y$, it holds that $|S_y|=f$.
Now,
\begin{align*}
|S_x-S_y|&=|S_x-S_x\cap S_y|=|S_x|-|S_x\cap S_y|=f-|S_x\cap S_y|\\
&=|S_y|-|S_x\cap S_y|=|S_y-S_x\cap S_y|\ge |S_y\cap A\pth{x}-S_x\cap S_y|\ge|S_y\cap A\pth{x}-S_x|.
\end{align*}
Thus, $|S_x-S_y|\ge |S_y\cap A\pth{x}-S_x|$.

By definition of $S_x$, for each $i\in S_x-S_y$ and $j\in S_y\cap A\pth{x}-S_x$, at point $x$, we have $h_i^{\prime}(x)\ge h_j^{\prime}(x)$, i.e., $h_i^{\prime}(x)\ge \max_{j\in S_y\cap A\pth{x}-S_x}h_j^{\prime}\pth{x}$. We have
\begin{align}
\nonumber
\sum_{i\in S_x-S_y}h_i^{\prime}\pth{y}-\sum_{i\in S_y\cap A\pth{x}-S_x}h_i^{\prime}\pth{x}
\nonumber
&\overset{(a)}{\ge} \sum_{i\in S_x-S_y}h_i^{\prime}\pth{x}-\sum_{i\in S_y\cap A\pth{x}-S_x}h_i^{\prime}\pth{x}\\
\nonumber
&\ge \sum_{i\in S_x-S_y}\max_{j\in S_y\cap A\pth{x}-S_x}h_j^{\prime}\pth{x}-\sum_{i\in S_y\cap A\pth{x}-S_x}h_i^{\prime}\pth{x}\\
\nonumber
&\overset{(b)}\ge \sum_{i\in  S_y\cap A\pth{x}-S_x}\max_{j\in S_y\cap A\pth{x}-S_x}h_j^{\prime}\pth{x}-\sum_{i\in S_y\cap A\pth{x}-S_x}h_i^{\prime}\pth{x}\\
&\ge 0,
\label{FcasesII}
\end{align}
where $(a)$ holds due to the fact that $h_i^{\prime}\pth{\cdot}$ is non-decreasing and that $y\ge x$,
and $(b)$ holds because $|S_x-S_y|\ge |S_y\cap A\pth{x}-S_x|$
and for $j\in S_y$, $h_j^{\prime}(x)>0$.

Therefore, from (\ref{Fcases}), (\ref{FcasesI}) and (\ref{FcasesII}), we have that $F(y)-F(x)\ge 0$ for $y\ge x$, i.e., $F(\cdot)$ is non-decreasing.\\
\vskip 2\baselineskip

\section*{$G\pth{\cdot}$ is non-decreasing}

Now we show that $G\pth{\cdot}$ is also non-decreasing. Choose any $x\in\mathbb{R}$, and
choose any $y\ge x$. Let $S_y$ and $S_x$ be sets such that $\sum_{i\in B\pth{y}-S_y}h_i^{\prime}\pth{y}$ and $\sum_{i\in B\pth{x}-S_x}h_i^{\prime}\pth{x}$ are maximized, respectively.

\begin{align}
\label{monotoneG}
\nonumber
G(y)-G(x)&=\sum_{i\in B\pth{y}-S_y}h_i^{\prime}\pth{y}-\sum_{i\in B\pth{x}-S_x}h_i^{\prime}\pth{x}\\
\nonumber
&=\sum_{i\in B\pth{y}-S_y-S_x}h_i^{\prime}\pth{y}+\sum_{i\in S_x\cap B\pth{y}-S_y}h_i^{\prime}\pth{y}-\pth{\sum_{i\in B\pth{x}-S_x-S_y}h_i^{\prime}\pth{x}+\sum_{i\in S_y\cap B\pth{x} -S_x}h_i^{\prime}\pth{x}}\\
\nonumber
&=\pth{\sum_{i\in B\pth{y}-S_y-S_x}h_i^{\prime}\pth{y}-\sum_{i\in B\pth{x}-S_x-S_y}h_i^{\prime}\pth{x}}+\pth{\sum_{i\in S_x\cap B\pth{y}-S_y}h_i^{\prime}\pth{y}-\sum_{i\in S_y\cap B\pth{x}-S_x}h_i^{\prime}\pth{x}}\\
\nonumber
&\overset{(a)}{\ge}\pth{\sum_{i\in B\pth{y}-S_y-S_x}h_i^{\prime}\pth{x}-\sum_{i\in B\pth{x}-S_x-S_y}h_i^{\prime}\pth{x}}+\pth{\sum_{i\in S_x\cap B\pth{y}-S_y}h_i^{\prime}\pth{y}-\sum_{i\in S_y\cap B\pth{x}-S_x}h_i^{\prime}\pth{x}}\\
\nonumber
&\overset{(b)}{\ge} \pth{\sum_{i\in B\pth{x}-S_y-S_x}h_i^{\prime}\pth{x}-\sum_{i\in B\pth{x}-S_x-S_y}h_i^{\prime}\pth{x}}+\pth{\sum_{i\in S_x\cap B\pth{y}-S_y}h_i^{\prime}\pth{y}-\sum_{i\in S_y\cap B\pth{x}-S_x}h_i^{\prime}\pth{x}}\\
&\overset{(c)}{=}0+\pth{\sum_{i\in S_x\cap B\pth{y}-S_y}h_i^{\prime}\pth{y}-\sum_{i\in S_y-S_x}h_i^{\prime}\pth{x}}.
\end{align}
Inequality $(a)$ holds due to the fact that $h_i^{\prime}\pth{\cdot}$ is non-decreasing and $x\le y$; inequality $(b)$ follows from the fact that $B\pth{y}\subseteq B\pth{x}$ and $h_i^{\prime}\pth{x}<0$ for each $i\in B\pth{x}$; and equality $(c)$ is true because that $S_y\subseteq B\pth{y}\subseteq B\pth{x}$.

Now consider two cases: (i) $|S_y|<f$ and (ii) $|S_y|=f$.
\paragraph{\bf Case (i): $|S_y|<f$.}
In this case, we have $S_y=B(y)$, and
\begin{align}
\nonumber
\sum_{i\in S_x\cap B\pth{y}-S_y}h_i^{\prime}\pth{y}-\sum_{i\in S_y-S_x}h_i^{\prime}\pth{x}&= \sum_{i\in \O}h_i^{\prime}\pth{y}-\sum_{i\in S_y-S_x}h_i^{\prime}\pth{x}\\
\nonumber
&= 0-\sum_{i\in S_y-S_x}h_i^{\prime}\pth{x}\\
&\ge 0.
\label{Gcase1}
\end{align}

\paragraph{\bf Case (ii): $|S_y|=f$.}
Because $S_y\subseteq B(y)\subseteq B(x)$, if $|S_y|=f$, we have $|B(x)|\ge f$. Then, by definition of $S_x$, it holds that $|S_x|=f$. Now
\begin{align*}
|S_y-S_x|&=|S_y-S_x\cap S_y|=|S_y|-|S_x\cap S_y|=f-|S_x\cap S_y|\\
&=|S_x|-|S_x\cap S_y|=|S_x-S_x\cap S_y|\ge |S_x\cap B\pth{y}-S_x\cap S_y|=|S_x\cap B\pth{y}-S_y|.
\end{align*}
Thus $|S_y-S_x|\ge |S_x\cap B\pth{y}-S_y|$.

By definition of $S_y$, for each $i\in S_y-S_x$ and $j\in S_x\cap B\pth{y}-S_y$, at point $y$, we have $h_i^{\prime}(y)\le \min_{j\in S_x\cap B\pth{y}-S_y}h_j^{\prime}\pth{y}$. We have
\begin{align}
\nonumber
\sum_{i\in S_x\cap B\pth{y}-S_y}h_i^{\prime}\pth{y}-\sum_{i\in S_y-S_x}h_i^{\prime}\pth{x}
&\overset{(a)}{\ge} \sum_{i\in S_x\cap B\pth{y}-S_y}h_i^{\prime}\pth{y}-\sum_{i\in S_y-S_x}h_i^{\prime}\pth{y}\\
\nonumber
&\ge \sum_{i\in S_x\cap B\pth{y}-S_y}h_i^{\prime}\pth{y}-\sum_{i\in S_y-S_x}\min_{j\in S_x\cap B\pth{y}-S_y}h_j^{\prime}\pth{y}\\
\nonumber
&\overset{(b)}{\ge} \sum_{i\in S_x\cap B\pth{y}-S_y}h_i^{\prime}\pth{y}-\sum_{i\in S_x\cap B\pth{y}-S_y}\min_{j\in S_x\cap B\pth{y}-S_y}h_j^{\prime}\pth{y}\\
&\ge 0,
\label{Gcase2}
\end{align}
where $(a)$ holds due to the fact that $h_i^{\prime}\pth{\cdot}$ is non-decreasing and that $y\ge x$, and $(b)$ holds because $|S_y-S_x|\ge |S_x\cap B\pth{y}-S_y|$ and for each $j\in B(y)$, $h_j^{\prime}(y)<0$.

Therefore, from (\ref{monotoneG}), (\ref{Gcase1}) and (\ref{Gcase2}), we have $G(y)-G(x)\ge 0$ for $y\ge x$, i.e., $G(\cdot)$ is non-decreasing.


\eproof
\end{proof}

\section{Proof of Proposition \ref{continuous}}
\label{appendix:continuous}

\begin{proof}
We first show that $F(x)$ is continuous. We will use the non-decreasing nature of $F(\cdot)$ proved above in Proposition \ref{monotone}.

Recall that each $h_i(x)$ is continuously differentiable,
 i.e., $h_i^{\prime}(x)$ is continuous. Then, for every $\epsilon>0$ there exists a $\delta>0$ such that for all $x\in (c-\delta, c+\delta)$ the following holds
{\em for all} $i\in\calN$,
\begin{align}
\label{continousofderivative}
|h_i^{\prime}(x)-h_i^{\prime}(c)|<\epsilon.
\end{align}
To show $F(x)$ is continuous, we need to show that
\begin{align}
|x-c|<\delta ~\Rightarrow ~ |F(x)-F(c)|<\epsilon.
\end{align}
Suppose $|x-c|<\delta $ holds for some $\delta>0$, then $c-\delta<x<c+\delta$. Let $S_{c+\delta}$ and $S_{c}$ be the subsets of $A\pth{c+\delta}$ and $A\pth{c}$, where $|S_{c+\delta}|\le f$ and $|S_{c}|\le f$, such that $\sum_{i\in A\pth{c+\delta}-S_{c+\delta}}h_i^{\prime}\pth{c+\delta}$ and $\sum_{i\in A\pth{c}-S_c}h_i^{\prime}\pth{c}$ are minimized, respectively. Note that $A(c)\subseteq A(c+\delta)$.

We have
\begin{align}
\nonumber
&F(x)-F(c)\overset{(a)}{\le} F\pth{c+\delta}-F\pth{c}\\
\nonumber
&\quad=\sum_{i\in A\pth{c+\delta}-S_{c+\delta}}h_i^{\prime}\pth{c+\delta}-\sum_{i\in A\pth{c}-S_{c}}h_i^{\prime}\pth{c}\\
\nonumber
&\quad=\sum_{i\in A\pth{c+\delta}-S_{c+\delta}-S_c}h_i^{\prime}\pth{c+\delta}+\sum_{i\in A\pth{c+\delta}\cap S_c-S_{c+\delta}}h_i^{\prime}\pth{c+\delta}-\pth{\sum_{i\in A\pth{c}-S_{c+\delta}-S_c}h_i^{\prime}\pth{c}+\sum_{i\in S_{c+\delta}\cap A\pth{c}-S_c}h_i^{\prime}\pth{c}}\\
\nonumber
&\quad\overset{(b)}{=}\sum_{i\in A\pth{c+\delta}-S_{c+\delta}-S_c}h_i^{\prime}\pth{c+\delta}+\sum_{i\in S_c-S_{c+\delta}}h_i^{\prime}\pth{c+\delta}-\pth{\sum_{i\in A\pth{c}-S_{c+\delta}-S_c}h_i^{\prime}\pth{c}+\sum_{i\in S_{c+\delta}\cap A\pth{c}-S_c}h_i^{\prime}\pth{c}}\\
\nonumber
&\quad= \pth{\sum_{i\in A\pth{c+\delta}-S_{c+\delta}-S_c}h_i^{\prime}\pth{c+\delta}-\sum_{i\in A\pth{c}-S_{c+\delta}-S_c}h_i^{\prime}\pth{c}}+\pth{\sum_{i\in S_c-S_{c+\delta}}h_i^{\prime}\pth{c+\delta}-\sum_{i\in S_{c+\delta}\cap A\pth{c}-S_c}h_i^{\prime}\pth{c}}\\
\nonumber
&\quad\overset{(c)}{\le} \pth{\sum_{i\in A\pth{c+\delta}-S_{c+\delta}-S_c}h_i^{\prime}\pth{c+\delta}-\sum_{i\in A\pth{c+\delta}-S_{c+\delta}-S_c}h_i^{\prime}\pth{c}}+\pth{\sum_{i\in S_c-S_{c+\delta}}h_i^{\prime}\pth{c+\delta}-\sum_{i\in S_{c+\delta}-S_c}h_i^{\prime}\pth{c}}\\
&\quad\overset{(d)}{\le} \pth{\sum_{i\in A\pth{c+\delta}-S_{c+\delta}-S_c}h_i^{\prime}\pth{c+\delta}-\sum_{i\in A\pth{c+\delta}-S_{c+\delta}-S_c}h_i^{\prime}\pth{c}}+\pth{\sum_{i\in S_{c+\delta}-S_c}h_i^{\prime}\pth{c+\delta}-\sum_{i\in S_{c+\delta}-S_c}h_i^{\prime}\pth{c}},
\label{e1}
\end{align}
where $(a)$ holds due to monotonicity of $F(\cdot)$; equality $(b)$ is true since $S_c\subseteq A(c)\subseteq A(c+\delta)$; inequality $(c)$ follows from the fact that $h_i^{\prime}\pth{c}\le 0$ for each $i\notin A\pth{c}$ and $A\pth{c}\subseteq A\pth{c+\delta}$; and inequality $(d)$
holds because, as shown next,
\begin{align}
\label{eee}
\sum_{i\in S_c-S_{c+\delta}}h_i^{\prime}\pth{c+\delta}\le \sum_{i\in S_{c+\delta}-S_c}h_i^{\prime}\pth{c+\delta}
\end{align}
Now, observing that $|S_{c}|\le |S_{c+\delta}|$, we get
\begin{align*}
| S_c-S_{c+\delta}|&=| S_c-S_c\cap S_{c+\delta}|=| S_c|-|S_c\cap S_{c+\delta}|\le |S_{c+\delta}|-|S_c\cap S_{c+\delta}|=|S_{c+\delta}-S_c|.
\end{align*}
 In addition, by definition of $S_c$, for each $i\in S_c-S_{c+\delta}$ and $j\in S_{c+\delta}-S_c$, $h_i^{\prime}\pth{c+\delta}\le h_j^{\prime}\pth{c+\delta}$. Then,
\begin{align*}
\sum_{i\in S_c-S_{c+\delta}}h_i^{\prime}\pth{c+\delta}&\le \sum_{i\in S_c-S_{c+\delta}}\min_{j\in S_{c+\delta}-S_c}h_j^{\prime}\pth{c+\delta}\\
&\overset{(a)}{\le} \sum_{i\in S_{c+\delta}-S_c}\min_{j\in S_{c+\delta}-S_c}h_j^{\prime}\pth{c+\delta}\\
&\le \sum_{i\in S_{c+\delta}-S_c}h_i^{\prime}\pth{c+\delta},
\end{align*}
where inequality $(a)$ is true because $\left | S_c-S_{c+\delta}\right |\le \left |S_{c+\delta}-S_c\right |$ and $\min_{j\in S_{c+\delta}-S_c}h_j^{\prime}\pth{c+\delta}>0$. This proves (\ref{eee}).

Then we have
\begin{align*}
F(x)-F(c)&\le \pth{\sum_{i\in A\pth{c+\delta}-S_{c+\delta}-S_c}h_i^{\prime}\pth{c+\delta}-\sum_{i\in A\pth{c+\delta}-S_{c+\delta}-S_c}h_i^{\prime}\pth{c}}\\
&\quad+\pth{\sum_{i\in S_{c+\delta}-S_c}h_i^{\prime}\pth{c+\delta}-\sum_{i\in S_{c+\delta}-S_c}h_i^{\prime}\pth{c}}\text{~~~~~~~~~~~~~~due to (\ref{e1})}\\
&\overset{(a)}= \sum_{i\in A\pth{c+\delta}-S_c}\pth{h_i^{\prime}\pth{c+\delta}-h_i^{\prime}\pth{c}}\\
&\overset{(b)}{<} \left | A\pth{c+\delta}-S_c\right |\epsilon\\
&< n\epsilon.
\end{align*}
Equality (a) follows
because $(A\pth{c+\delta}-S_{c+\delta}-S_c)\cup (S_{c+\delta}-S_c)=A\pth{c+\delta}-S_c$ and sets $A\pth{c+\delta}-S_{c+\delta}-S_c$ and $S_{c+\delta}-S_c$ are disjoint. Inequality $(b)$ follows from (\ref{continousofderivative}).

By an analogous argument, we can also show that
for any $x\in (c-\delta, c+\delta)$, $$F(x)-F(c)> -n\epsilon.$$ For completeness, we present the proof as follows.


  Let $S_{c-\delta}$ and $S_{c}$ be the subsets of $A\pth{c-\delta}$ and $A\pth{c}$, where $|S_{c-\delta}|\le f$ and $|S_{c}|\le f$, such that $\sum_{i\in A\pth{c-\delta}-S_{c-\delta}}h_i^{\prime}\pth{c-\delta}$ and $\sum_{i\in A\pth{c}-S_c}h_i^{\prime}\pth{c}$ are minimized, respectively.

\begin{align*}
&F(x)-F(c)\ge F(c-\delta)-F(c)\\
&=\sum_{i\in A\pth{c-\delta}-S_{c-\delta}}h_i^{\prime}\pth{c-\delta}-\sum_{i\in A\pth{c}-S_c}h_i^{\prime}\pth{c}\\
&=\pth{\sum_{i\in A\pth{c-\delta}-S_{c-\delta}-S_c}h_i^{\prime}\pth{c-\delta}+\sum_{i\in S_c\cap A\pth{c-\delta}-S_{c-\delta}}h_i^{\prime}\pth{c-\delta}}-\pth{\sum_{i\in A\pth{c}-S_{c-\delta}-S_c}h_i^{\prime}\pth{c}+\sum_{i\in S_{c-\delta}\cap A(c)-S_c}h_i^{\prime}\pth{c}}\\
&=\pth{\sum_{i\in A\pth{c-\delta}-S_{c-\delta}-S_c}h_i^{\prime}\pth{c-\delta}-\sum_{i\in A\pth{c}-S_{c-\delta}-S_c}h_i^{\prime}\pth{c}}+\pth{\sum_{i\in S_c\cap A\pth{c-\delta}-S_{c-\delta}} h_i^{\prime}\pth{c-\delta}-\sum_{i\in S_{c-\delta}\cap A(c)-S_c}h_i^{\prime}\pth{c}}\\
&\overset{(a)}{\ge} \pth{\sum_{i\in A\pth{c}-S_{c-\delta}-S_c}h_i^{\prime}\pth{c-\delta}-\sum_{i\in A\pth{c}-S_{c-\delta}-S_c}h_i^{\prime}\pth{c}}+\pth{\sum_{i\in S_c-S_{c-\delta}} h_i^{\prime}\pth{c-\delta}-\sum_{i\in S_{c-\delta}\cap A(c)-S_c}h_i^{\prime}\pth{c}}\\
&\overset{(b)}{=}\pth{\sum_{i\in A\pth{c}-S_{c-\delta}-S_c}h_i^{\prime}\pth{c-\delta}-\sum_{i\in A\pth{c}-S_{c-\delta}-S_c}h_i^{\prime}\pth{c}}+\pth{\sum_{i\in S_c-S_{c-\delta}} h_i^{\prime}\pth{c-\delta}-\sum_{i\in S_{c-\delta}-S_c}h_i^{\prime}\pth{c}}\\
&=\sum_{i\in A\pth{c}-S_{c-\delta}-S_c}\pth{h_i^{\prime}\pth{c-\delta}-h_i^{\prime}\pth{c}}+\pth{\sum_{i\in S_c-S_{c-\delta}} h_i^{\prime}\pth{c-\delta}-\sum_{i\in S_{c-\delta}-S_c}h_i^{\prime}\pth{c}}.
\end{align*}
Inequality $(a)$ follows from the fact that $h_i^{\prime}\pth{c-\delta}\le 0$ for each $i\notin A\pth{c-\delta}$ and $A\pth{c-\delta}\subseteq A(c)$. Equality $(b)$ is true because that $S_{c-\delta}\subseteq A(c-\delta)\subseteq A(c)$. Now, observing that $|S_{c-\delta}|\le |S_c|$, we get
 \begin{align}
 \label{lbF}
 |S_c-S_{c-\delta}|=|S_c|-|S_c\cap S_{c-\delta}|\ge |S_{c-\delta}|-|S_c\cap S_{c-\delta}|=|S_{c-\delta}-S_{c}|.
 \end{align}

In addition, we have
\begin{align}
\label{lb continuous F}
\nonumber
\sum_{i\in S_c-S_{c-\delta}} h_i^{\prime}\pth{c-\delta}-\sum_{i\in S_{c-\delta}-S_c}h_i^{\prime}\pth{c}&
\overset{(a)}{\ge} \sum_{i\in S_c-S_{c-\delta}} h_i^{\prime}\pth{c-\delta}-\sum_{i\in S_{c-\delta}-S_c}\min_{j\in S_c-S_{c-\delta}}h_j^{\prime}\pth{c}\\
\nonumber
&\overset{(b)}{\ge} \sum_{i\in S_c-S_{c-\delta}} h_i^{\prime}\pth{c-\delta}-\sum_{i\in S_c-S_{c-\delta}}\min_{j\in S_c-S_{c-\delta}}h_j^{\prime}\pth{c}\\
\nonumber
&\ge \sum_{i\in S_c-S_{c-\delta}} h_i^{\prime}\pth{c-\delta}-\sum_{i\in S_c-S_{c-\delta}}h_i^{\prime}\pth{c}\\
&=\sum_{i\in S_c-S_{c-\delta}} \pth{h_i^{\prime}\pth{c-\delta}-h_i^{\prime}\pth{c}}.
\end{align}
Inequality $(a)$ holds due to the fact that for each $i\in S_{c-\delta}-S_c$, $h_i^{\prime}(c)\le \min_{j\in S_c-S_{c-\delta}}h_j^{\prime}\pth{c}$. Inequality $(b)$ follows from (\ref{lbF}) and the fact that $\min_{j\in S_c-S_{c-\delta}}h_j^{\prime}\pth{c}>0$.

Thus
\begin{align*}
F(x)-F(c)&\ge \sum_{i\in A\pth{c}-S_{c-\delta}-S_c}\pth{h_i^{\prime}\pth{c-\delta}-h_i^{\prime}\pth{c}}+\pth{\sum_{i\in S_c-S_{c-\delta}} h_i^{\prime}\pth{c-\delta}-\sum_{i\in S_{c-\delta}-S_c}h_i^{\prime}\pth{c}}\\
&\ge \sum_{i\in A\pth{c}-S_{c-\delta}-S_c}\pth{h_i^{\prime}\pth{c-\delta}-h_i^{\prime}\pth{c}}
+\sum_{i\in S_c-S_{c-\delta}} \pth{h_i^{\prime}\pth{c-\delta}-h_i^{\prime}\pth{c}}~~~\text{from (\ref{lb continuous F})}\\
&=\sum_{i\in A\pth{c}-S_{c-\delta}}\pth{h_i^{\prime}\pth{c-\delta}-h_i^{\prime}\pth{c}}\\
&> -n\epsilon ~~~~\text{from (\ref{continousofderivative})}
\end{align*}

\noindent
Then we have, for any $\epsilon_0=n\epsilon>0$, there exists $\delta>0$ such that
\begin{align*}
|x-c|<\delta ~\Rightarrow ~ |F(x)-F(c)|<\epsilon_0.
\end{align*}
Therefore, $F\pth{\cdot}$ is continuous.

%


\section*{Continuity of $G\pth{\cdot}$}
\label{appendix:continuity:G}

To show $G(\cdot)$ is continuous, we need to show that
\begin{align}
|x-c|<\delta ~\Rightarrow ~ |G(x)-G(c)|<\epsilon.
\end{align}
Suppose $|x-c|<\delta $ holds for some $\delta>0$, then $c-\delta<x<c+\delta$. Let $S_{c+\delta}$ and $S_{c}$ be the subsets of $B\pth{c+\delta}$ and $B\pth{c}$, where $|S_{c+\delta}|\le f$ and $|S_{c}|\le f$, such that $\sum_{i\in B\pth{c+\delta}-S_{c+\delta}}h_i^{\prime}\pth{c+\delta}$ and $\sum_{i\in B\pth{c}-S_c}h_i^{\prime}\pth{c}$ are maximized, respectively.

We have
\begin{align*}
&G(x)-G(c)\le G\pth{c+\delta}-G\pth{c}~~\text{because $G(\cdot)$ is non-decreasing}\\
&=\sum_{i\in B\pth{c+\delta}-S_{c+\delta}}h_i^{\prime}\pth{c+\delta}-\sum_{i\in B\pth{c}-S_{c}}h_i^{\prime}\pth{c}\\
&=\sum_{i\in B\pth{c+\delta}-S_{c+\delta}-S_c}h_i^{\prime}\pth{c+\delta}+\sum_{i\in S_c\cap B\pth{c+\delta}-S_{c+\delta}}h_i^{\prime}\pth{c+\delta}-\pth{\sum_{i\in B\pth{c}-S_{c+\delta}-S_c}h_i^{\prime}\pth{c}+\sum_{i\in S_{c+\delta}\cap B\pth{c}-S_c}h_i^{\prime}\pth{c}}\\
&= \pth{\sum_{i\in B\pth{c+\delta}-S_{c+\delta}-S_c}h_i^{\prime}\pth{c+\delta}-\sum_{i\in B\pth{c}-S_{c+\delta}-S_c}h_i^{\prime}\pth{c}}+\pth{\sum_{i\in S_c\cap B\pth{c+\delta}-S_{c+\delta}}h_i^{\prime}\pth{c+\delta}-\sum_{i\in S_{c+\delta}\cap B\pth{c}-S_c}h_i^{\prime}\pth{c}}\\
&\overset{(a)}{\le} \pth{\sum_{i\in B\pth{c}-S_{c+\delta}-S_c}h_i^{\prime}\pth{c+\delta}-\sum_{i\in B\pth{c}-S_{c+\delta}-S_c}h_i^{\prime}\pth{c}}+\pth{\sum_{i\in S_c-S_{c+\delta}}h_i^{\prime}\pth{c+\delta}-\sum_{i\in S_{c+\delta}-S_c}h_i^{\prime}\pth{c}}\\
&=\sum_{i\in B\pth{c}-S_{c+\delta}-S_c}\pth{h_i^{\prime}\pth{c+\delta}-h_i^{\prime}\pth{c}}+\pth{\sum_{i\in S_c-S_{c+\delta}}h_i^{\prime}\pth{c+\delta}-\sum_{i\in S_{c+\delta}-S_c}h_i^{\prime}\pth{c}},
\end{align*}
where inequality $(a)$ follows from the fact that $h_i^{\prime}\pth{c+\delta}\ge 0$ for each $i\notin B\pth{c+\delta}$ and $B\pth{c}\supseteq B\pth{c+\delta}\supseteq S_{c+\delta}$. Next we show
\begin{align}
\label{GpassUp}
\sum_{i\in S_{c+\delta}-S_c}h_i^{\prime}\pth{c}\ge \sum_{i\in S_c-S_{c+\delta}}h_i^{\prime}\pth{c}.
\end{align}
For each $i\notin S_c$, it holds that $h_i^{\prime}(c)\ge \max_{j\in S_c-S_{c+\delta}} h_j^{\prime}(c)$. Now, observing that $|S_c|\ge |S_{c+\delta}|$, we get
$$|S_c-S_{c+\delta}|=|S_c|-|S_c\cap S_{c+\delta}|\ge |S_{c+\delta}|-|S_c\cap S_{c+\delta}|=|S_{c+\delta}-S_{c}|.$$
Thus, because $\max_{j\in S_c-S_{c+\delta}} h_j^{\prime}(c)<0$,
\begin{align*}
\sum_{i\in S_{c+\delta}-S_c}h_i^{\prime}\pth{c}&\ge \sum_{i\in S_{c+\delta}-S_c}\max_{j\in S_c-S_{c+\delta}} h_j^{\prime}(c)\ge \sum_{i\in S_c-S_{c+\delta}}\max_{j\in S_c-S_{c+\delta}} h_j^{\prime}(c)\ge \sum_{i\in S_c-S_{c+\delta}} h_i^{\prime}(c).
\end{align*}

So we have
\begin{align*}
G(x)-G(c)&\le
\sum_{i\in B\pth{c}-S_{c+\delta}-S_c}\pth{h_i^{\prime}\pth{c+\delta}-h_i^{\prime}\pth{c}}+\pth{\sum_{i\in S_c-S_{c+\delta}}h_i^{\prime}\pth{c+\delta}-\sum_{i\in S_c-S_{c+\delta}}h_i^{\prime}\pth{c}}\\
&=\sum_{i\in B\pth{c}-S_{c+\delta}}\pth{h_i^{\prime}\pth{c+\delta}-h_i^{\prime}\pth{c}}\\
&< n\epsilon ~~~~\text{due to (\ref{continousofderivative})}
\end{align*}
\vskip 2\baselineskip

Now we show that for any $x\in (c-\delta, c+\delta)$, it follows that $G(x)-G(c)\ge -n\epsilon$.  Let $S_{c-\delta}$ and $S_{c}$ be the subsets of $B\pth{c-\delta}$ and $B\pth{c}$, where $|S_{c-\delta}|\le f$ and $|S_{c}|\le f$, such that $\sum_{i\in B\pth{c-\delta}-S_{c-\delta}}h_i^{\prime}\pth{c-\delta}$ and $\sum_{i\in B\pth{c}-S_c}h_i^{\prime}\pth{c}$ are maximized, respectively. Note that $B(c-\delta)\supseteq B(c)\supseteq S_c$.

\begin{align*}
&G(x)-G(c)\ge G(c-\delta)-G(c)~~\text{because $G(\cdot)$ is non-decreasing}\\
&=\sum_{i\in B\pth{c-\delta}-S_{c-\delta}}h_i^{\prime}\pth{c-\delta}-\sum_{i\in B\pth{c}-S_c}h_i^{\prime}\pth{c}\\
&=\pth{\sum_{i\in B\pth{c-\delta}-S_{c-\delta}-S_c}h_i^{\prime}\pth{c-\delta}+\sum_{i\in S_c\cap B\pth{c-\delta}-S_{c-\delta}}h_i^{\prime}\pth{c-\delta}}-\pth{\sum_{i\in B\pth{c}-S_{c-\delta}-S_c}h_i^{\prime}\pth{c}+\sum_{i\in S_{c-\delta}\cap B\pth{c} -S_c}h_i^{\prime}\pth{c}}\\
&=\pth{\sum_{i\in B\pth{c-\delta}-S_{c-\delta}-S_c}h_i^{\prime}\pth{c-\delta}+\sum_{i\in S_c-S_{c-\delta}}h_i^{\prime}\pth{c-\delta}}-\pth{\sum_{i\in B\pth{c}-S_{c-\delta}-S_c}h_i^{\prime}\pth{c}+\sum_{i\in S_{c-\delta}\cap B\pth{c} -S_c}h_i^{\prime}\pth{c}}\\
&=\pth{\sum_{i\in B\pth{c-\delta}-S_{c-\delta}-S_c}h_i^{\prime}\pth{c-\delta}-\sum_{i\in B\pth{c}-S_{c-\delta}-S_c}h_i^{\prime}\pth{c}}+\pth{\sum_{i\in S_c-S_{c-\delta}} h_i^{\prime}\pth{c-\delta}-\sum_{i\in S_{c-\delta}\cap B\pth{c}-S_c}h_i^{\prime}\pth{c}}\\
&\overset{(a)}{\ge} \pth{\sum_{i\in B\pth{c-\delta}-S_{c-\delta}-S_c}h_i^{\prime}\pth{c-\delta}-\sum_{i\in B\pth{c-\delta}-S_{c-\delta}-S_c}h_i^{\prime}\pth{c}}+\pth{\sum_{i\in S_c-S_{c-\delta}} h_i^{\prime}\pth{c-\delta}-\sum_{i\in S_{c-\delta}-S_c}h_i^{\prime}\pth{c}}\\
&=\sum_{i\in B\pth{c-\delta}-S_{c-\delta}-S_c}\pth{h_i^{\prime}\pth{c-\delta}-h_i^{\prime}\pth{c}}+\pth{\sum_{i\in S_c-S_{c-\delta}} h_i^{\prime}\pth{c-\delta}-\sum_{i\in S_{c-\delta}-S_c}h_i^{\prime}\pth{c}},
\end{align*}
where inequality $(a)$ follows because $B(c)\subseteq B(c-\delta)$ and for each $i\notin B(c)$, $h_i^{\prime}(c)\ge 0$. Now, observing that $|S_c|\le |S_{c-\delta}|$, we get
$$|S_c-S_{c-\delta}|=|S_c|-|S_c\cap S_{c-\delta}|\le |S_{c-\delta}|-|S_c\cap S_{c-\delta}|=|S_{c-\delta}-S_c|,$$
and for each $i\notin S_{c-\delta}$,
$$h_i^{\prime}(c-\delta)\ge \max_{j\in S_{c-\delta}-S_c} h_j^{\prime}(c-\delta).$$
Thus,
\begin{align*}
G(x)-G(c)&\ge \sum_{i\in B\pth{c-\delta}-S_{c-\delta}-S_c}\pth{h_i^{\prime}\pth{c-\delta}-h_i^{\prime}\pth{c}}+\pth{\sum_{i\in S_c-S_{c-\delta}} h_i^{\prime}\pth{c-\delta}-\sum_{i\in S_{c-\delta}-S_c}h_i^{\prime}\pth{c}}\\
&\ge \sum_{i\in B\pth{c-\delta}-S_{c-\delta}-S_c}\pth{h_i^{\prime}\pth{c-\delta}-h_i^{\prime}(c)}+\pth{\sum_{i\in S_c-S_{c-\delta}} \max_{j\in S_{c-\delta}-S_c}h_j^{\prime}(c-\delta)-\sum_{i\in S_{c-\delta}-S_c}h_i^{\prime}\pth{c}}\\
&\overset{(a)}{\ge} \sum_{i\in B\pth{c-\delta}-S_{c-\delta}-S_c}\pth{h_i^{\prime}\pth{c-\delta}-h_i^{\prime}(c)}+\pth{\sum_{i\in S_{c-\delta}-S_c} \max_{j\in S_{c-\delta}-S_c}h_j^{\prime}(c-\delta)-\sum_{i\in S_{c-\delta}-S_c}h_i^{\prime}\pth{c}}\\
&\ge \sum_{i\in B\pth{c-\delta}-S_{c-\delta}-S_c}\pth{h_i^{\prime}\pth{c-\delta}-h_i^{\prime}(c)}+\pth{\sum_{i\in S_{c-\delta}-S_c} h_i^{\prime}(c-\delta)-\sum_{i\in S_{c-\delta}-S_c}h_i^{\prime}\pth{c}}\\
&=\sum_{i\in B\pth{c-\delta}-S_c}\pth{h_i^{\prime}\pth{c-\delta}-h_i^{\prime}(c)}\\
&>-n\epsilon ~~~~\text{due to (\ref{continousofderivative})}
\end{align*}
Inequality $(a)$ follows because $|S_{c-\delta}-S_c|\ge |S_c-S_{c-\delta}|$ and $\max_{j\in S_{c-\delta}-S_c}h_j^{\prime}(c-\delta)<0$.

$h_i^{\prime}(c)\ge 0$ for each $i\notin B(c)$ and that $B(c)\subseteq B(c-\delta)$.

Thus we have, for any $\epsilon_0=n\epsilon$, there exists $\delta$ such that
\begin{align*}
|x-c|<\delta ~\Rightarrow ~ |G(x)-G(c)|<\epsilon_0.
\end{align*}
Therefore, $G\pth{\cdot}$ is continuous.

\eproof
\end{proof}

\section{Proof of Proposition \ref{Continous and Monotone of the K-th Rank Functions}} \label{Proof of Continous and Monotone of the K-th Rank Functions}

\begin{proof}
We first show that $g_K(x)$ is non-decreasing. \\

For $1\le K\le n$, define $C\pth{x, K}\subseteq \calV$ as a collection of agents such that $|C\pth{x, K}|=K$ and
\[
h_i^{\prime}(x)\ge h_j^{\prime}(x),
\]
for each $i\in C\pth{x, K}$ and for each  $j\notin C\pth{x, K}$.

Let $x_1, x_2\in \reals$ such that $x_1<x_2$. To show that $g_K(x)$ is non-decreasing, we need to show that $g_K(x_1)\le g_K(x_2)$. Suppose, on the contrary, that $g_K(x_1)> g_K(x_2)$. That is, $\min_{k\in C\pth{x_1, K}}h_k^{\prime}\pth{x_1}>\min_{k\in C\pth{x_2, K}}h_k^{\prime}\pth{x_2}$, by definition of sets $C\pth{x_1, K}$ and $C\pth{x_2, K}$. Letting $k_1\in \argmin_{k\in C\pth{x_1, K}}h_k^{\prime}\pth{x_1}$ and $k_2\in \argmin_{k\in C\pth{x_2, K}}h_k^{\prime}\pth{x_2}$, the previous consequence can be rewritten as $h_{k_2}^{\prime}\pth{x_2}<h_{k_1}^{\prime}\pth{x_1}$ and $k_1\not=k_2$. The claim that $k_1\not=k_2$ follows from the fact that $h_i^{\prime}$ is non-decreasing for each $i$.

Thus, for each $k\in C\pth{x_1, K}$, we have
\begin{align}
\nonumber
h_{k_2}^{\prime}\pth{x_2}<h_{k_1}^{\prime}(x_1)&\le h_k^{\prime}(x_1)\quad \text{by definition of $k_1$}\\
&\le h_k^{\prime}(x_2) \quad \text{since $h_k^{\prime}\pth{\cdot}$ is non-decreasing}
\label{rank i}
\end{align}
Thus $h_{k_2}^{\prime}\pth{x_2}<h_k^{\prime}\pth{x_2}$ for each $k\in C\pth{x_1, K}$. By definition of $C\pth{x_2, K}$, it follows that $C\pth{x_1, K}\subseteq C\pth{x_2, K}$. In addition, it must be that $k_2\notin C\pth{x_1, K}$; otherwise, by (\ref{rank i}) we have $h_{k_2}^{\prime}\pth{x_2}<h_{k_2}^{\prime}\pth{x_2}$, a contradiction. Thus $C\pth{x_1, K}\cup \{k_2\}\subseteq C\pth{x_2, K}$.

Recall that $|C\pth{x, K}|=K$ for any $x$ and any $1\le K\le n$. Then,
\begin{align}
\label{cardinality}
K=|C\pth{x_2, K}|\ge |C\pth{x_1, K}\cup \{k_2\}|=|C\pth{x_1, K}|+|\{k_2\}|=K+1,
\end{align}
a contradiction.

Therefore, $g_K(x_1)\le g_K(x_2)$ and $g_K\pth{\cdot}$ is non-decreasing. \\

Recall that each $h_i(x)$ is continuously differentiable.  Thus function $h_i^{\prime}(x)$ exists and is continuous, i.e., for all $i$ and any $\epsilon>0$, there exists $\delta>0$ such that
\[
x\in (c-\delta, c+\delta) \Rightarrow ~ |h_i^{\prime}(x)-h_i^{\prime}(c)|\le \epsilon.
\]
Now we show that function $g_K\pth{\cdot}$ is continuous. In particular, we show that  $\forall~\epsilon>0$, $\exists~ \delta>0$ such that
\begin{align*}
|x-c|<\delta \Rightarrow |g_K\pth{x}-g_K\pth{c}|<\epsilon.
\end{align*}

\vskip 2\baselineskip

Let $k_1\in \argmin_{k\in C\pth{c, K}}h_k^{\prime}\pth{c}$, $k_2\in \argmin_{k\in C\pth{c+\delta, K}}h_k^{\prime}\pth{c+\delta}$ and $k_3\in \argmin_{k\in C\pth{c-\delta, K}}h_k^{\prime}\pth{c-\delta}$. We first prove that $x<c+\delta \Rightarrow g_K\pth{x}< g_K\pth{c}+\epsilon$. We consider two cases: (i) $h_{k_1}^{\prime}\pth{c+\delta}\ge h_{k_2}^{\prime}\pth{c+\delta}$ and (ii) $h_{k_1}^{\prime}\pth{c+\delta}< h_{k_2}^{\prime}\pth{c+\delta}$, respectively.

\paragraph{{\bf Case 1}: $h_{k_1}^{\prime}\pth{c+\delta}\ge h_{k_2}^{\prime}\pth{c+\delta}$.}
\begin{align*}
g_K\pth{x}&\le g_K\pth{c+\delta}\quad \text{by monotonicity of $g_K\pth{\cdot}$}\\
&=h_{k_2}^{\prime}\pth{c+\delta} \quad \text{by definition of $g_K\pth{c+\delta}$}\\
&\le h_{k_1}^{\prime}\pth{c+\delta} \quad \text{by assumption}\\
&< h_{k_1}^{\prime}\pth{c}+\epsilon \quad \text{by continuity of $h_{k_1}^{\prime}\pth{\cdot}$}\\
&=g_K\pth{c}+\epsilon \quad \text{by definition of $g_K\pth{c}$}.
\end{align*}

\paragraph{{\bf Case 2}: $h_{k_1}^{\prime}\pth{c+\delta}< h_{k_2}^{\prime}\pth{c+\delta}$.}
We observe that there exists $k^*$ such that $k^*\notin C\pth{c, K}$ but $k^*\in C\pth{c+\delta, K}$. Note that it is possible that $k^*=k_2$. If this is not true, then for each $k\in C\pth{c+\delta, K}$, it also holds that $k\in C\pth{c, K}$, i.e., $C\pth{c+\delta, K}\subseteq C\pth{c, K}$. On the other hand, we know $k_1\in C\pth{c, K}$ but $k_1\notin C\pth{c+\delta, K}$, which follows by the assumption of case 2 and the definition of $k_2$. Thus we have $C\pth{c+\delta, K}\cup \{k_1\}\subseteq C\pth{c, K}$. Similar as (\ref{cardinality}), we will arrive at a contradiction, and the claim follows.

With this observation, we have
\begin{align*}
g_K\pth{x}&\le g_K\pth{c+\delta}\quad \text{by monotonicity of $g_K\pth{\cdot}$}\\
&=h_{k_2}^{\prime}\pth{c+\delta} \quad \text{by definition of $g_K\pth{c+\delta}$}\\
&\le h_{k^*}^{\prime}\pth{c+\delta} \quad \text{by the fact that $k^*\in C\pth{c+\delta, K}$}\\
&< h_{k^*}^{\prime}\pth{c}+\epsilon \quad \text{by continuity of $h_{k^*}^{\prime}\pth{\cdot}$}\\
&\le h_{k_1}^{\prime}\pth{c}+\epsilon \quad \text{by the fact that $k^*\notin C\pth{c, K}$}\\
&=g_K\pth{c}+\epsilon \quad \text{by definition of $g_K\pth{c}$}.
\end{align*}

Thus, we have shown that $x<c+\delta \Rightarrow g_K\pth{x}< g_K\pth{c}+\epsilon$.
\vskip 2\baselineskip

Now we show that $x>c-\delta \Rightarrow g_K\pth{x}> g_K\pth{c}-\epsilon$. Recall that $k_1\in \argmin_{k\in C\pth{c, K}}h_k^{\prime}\pth{c}$ and $k_3\in \argmin_{k\in C\pth{c-\delta, K}}h_k^{\prime}\pth{c-\delta}$. We consider two cases: (i) $h_{k_1}^{\prime}\pth{c-\delta}\le h_{k_3}^{\prime}\pth{c-\delta}$ and (ii) $h_{k_1}^{\prime}\pth{c-\delta}> h_{k_3}^{\prime}\pth{c-\delta}$, respectively.

\paragraph{{\bf Case 1}: $h_{k_1}^{\prime}\pth{c-\delta}\le h_{k_3}^{\prime}\pth{c-\delta}$.}
\begin{align*}
g_K\pth{x}&\ge g_K\pth{c-\delta}\quad \text{by monotonicity of $g_K\pth{\cdot}$}\\
&=h_{k_3}^{\prime}\pth{c-\delta} \quad \text{by definition of $g_K\pth{c-\delta}$}\\
&\ge h_{k_1}^{\prime}\pth{c-\delta} \quad \text{by assumption}\\
&> h_{k_1}^{\prime}\pth{c}-\epsilon \quad \text{by continuity of $h_{k_1}^{\prime}\pth{\cdot}$}\\
&=g_K\pth{c}-\epsilon \quad \text{by definition of $g_K\pth{c}$}.
\end{align*}

\paragraph{{\bf Case 2}: $h_{k_1}^{\prime}\pth{c-\delta}> h_{k_3}^{\prime}\pth{c-\delta}$.}

If $k_3\in C\pth{c, K}$,
\begin{align*}
g_K\pth{x}&\ge g_K\pth{c-\delta}\quad \text{by monotonicity of $g_K\pth{\cdot}$}\\
&=h_{k_3}^{\prime}\pth{c-\delta} \quad \text{by definition of $g_K\pth{c-\delta}$}\\
&> h_{k_3}^{\prime}\pth{c}-\epsilon \quad \text{by continuity of $h_{k_3}^{\prime}\pth{\cdot}$}\\
&\ge h_{k_1}^{\prime}\pth{c}-\epsilon \quad \text{by definition of $C\pth{c, K}$ and the fact that $k_3\in C\pth{c, K}$}\\
&=g_K\pth{c}-\epsilon \quad \text{by definition of $g_K\pth{c}$}.
\end{align*}

If $k_3\notin C\pth{c, K}$, then there exists $k^*$ such that $k^*\notin C(c-\delta, K)$ and $k^*\in C(c, K)$. If this is not true, then for each $k$ such that $k\in C(c, K)$, it also holds that $k\in C(c-\delta, K)$, i.e., $C\pth{c, K}\subseteq C\pth{c-\delta, K}$.  By assumption, $k_3\notin C\pth{c, K}$. Thus, we get $C\pth{c, K}\cup \{k_3\}\subseteq C\pth{c-\delta, K}.$ Similar as (\ref{cardinality}), we will arrive at a contradiction, and the claim follows.

With this observation, we have
\begin{align*}
g_K\pth{x}&\ge g\pth{c-\delta}\quad \text{by monotonicity of $g_K\pth{\cdot}$}\\
&=h_{k_3^{\prime}}\pth{c-\delta} \quad \text{by definition of $g_K\pth{c-\delta}$}\\
&\ge h_{k^*}^{\prime}\pth{c-\delta} \quad \text{by the fact that $k^*\notin C\pth{c-\delta, K}$}\\
&> h_{k^*}^{\prime}\pth{c}-\epsilon \quad \text{by continuity of $h_{k^*}^{\prime}\pth{\cdot}$}\\
&\ge h_{k_1}^{\prime}\pth{c}-\epsilon \quad \text{by the fact that $k^*\in C\pth{c, K}$}\\
&=g_K\pth{c}-\epsilon \quad \text{by definition of $g_K\pth{c}$}.
\end{align*}

Thus, $x>c-\delta \Rightarrow g_K\pth{x}> g_K\pth{c}-\epsilon$.\\

Therefore, we have shown that  $\forall~\epsilon>0$, $\exists~ \delta>0$ such that
\begin{align*}
|x-c|<\delta \Rightarrow |g_K\pth{x}-g_K\pth{c}|<\epsilon,
\end{align*}
i.e., $g_K\pth{\cdot}$ is continuous.
\eproof
\end{proof}

\section{Proof of Lemma \ref{alternative well-definedOutPut}}\label{appendix: alternative well-definedOutPut}
\begin{proof}
If there exists $x\in\mathbb{R}$ that satisfies equation (\ref{alg1.3}) in
Algorithm 2, then the algorithm will not return $\bot$. Thus,
to prove this lemma, it suffices to show that there exists $x\in\mathbb{R}$ that satisfies equation (\ref{alg1.3}).
Consider the multiset of admissible functions $\{h_1(x),h_2(x),\cdots,h_n(x)\}$ obtained by a non-faulty
agent in Step 1 of Algorithm 2.
Define $X_i =\arg\max_{x\in\mathbb{R}} h_i(x)$.
Let $\max X_i$ and $\min X_i$ denote the largest and smallest values in $X_i$, respectively.
Sort the above $n$ functions $h_i(x)$ in an {\em increasing} order of their $\max X_i$ values, breaking
ties arbitrarily. Let $i_0$ denote the $f+1$-th agent in this sorted order (i.e., $i_0$ has
the $f+1$-th smallest value in the above sorted order).
Similarly,
sort the functions $h_i(x)$ in an {\em decreasing} order of $\min X_i$ values, breaking
ties arbitrarily. Let $j_0$ denote the $f+1$-th agent in this sorted order (i.e., $j_0$ has
the $f+1$-th largest value in the above sorted order).

Consider $x_1\in X_{i_0}$ and $x_2\in X_{j_0}$. By the choice of $x_1$ and the definition of $i_0$, at most $f$ values in $\{h_1^{\prime}(x_1), h_2^{\prime}(x_1), \cdots, h_n^{\prime}(x_1)\}$ can be positive. Recall that $$g_K(x)\triangleq K^{th} ~\text{largest value in the set}~ \{h_1^{\prime}(x), h_2^{\prime}(x), \cdots, h_n^{\prime}(x)\}.$$ Thus we have $g_K(x_1)\le 0$ for $K=f+1, \ldots, n$. Consequently, $\sum_{K=f+1}^{n-f} g_K\pth{x_1}\le 0$.

Similarly, it can be shown that $g_K(x_2)\ge 0$ for $k=1,\ldots, n-f$, and $\sum_{K=f+1}^{n-f} g_K\pth{x_2}\ge 0$.\\

If $\sum_{K=f+1}^{n-f} g_K\pth{x_1}= 0$ or $\sum_{K=f+1}^{n-f} g_K\pth{x_2}=0$, then $x_1$ or $x_2$, respectively,
satisfy equation (\ref{alg1.3}), proving the lemma.

Let us now consider the case when $\sum_{K=f+1}^{n-f} g_K\pth{x_1}< 0$ and $\sum_{K=f+1}^{n-f} g_K\pth{x_2}>0$.
By Proposition \ref{Continous and Monotone of the K-th Rank Functions},
we have that $\sum_{K=f+1}^{n-f} g_K\pth{\cdot}$ is non-decreasing and continuous. Then it follows that $x_1< x_2$, and there exists $\tx\in [x_1, x_2]$ such that
$\sum_{K=f+1}^{n-f} g_K\pth{\tx}=0$, i.e., $\tx$ satisfies equation (\ref{alg1.3}), proving the lemma.

\eproof
\end{proof}

\section{Proofs of Theorems \ref{t_algo2} and \ref{t_algo2-1}}\label{appendix:t_algo2}

By Lemma \ref{alternative well-definedOutPut}, we know that Algorithm $2$ returns a value in $\mathbb{R}$.
Let $\tx$ be the output of Algorithm $2$ for the set of functions $\{h_1(x),h_2(x),\cdots,h_n(x)\}$
gathered in Step 1 of the algorithm.
Sort the above $n$ functions $h_i(x)$ in a {\em non-increasing} order of their $h_i^{\prime}(\tx)$ values, breaking
ties arbitrarily. Let $F_1^*$ denote the first $f$ agents in this sorted order (i.e., agents in $F_1^*$ have the $f$ largest values in the above sorted order); and let $F_2^*$ denote the last $f$ agents in this sorted order (i.e., agents in $F_2^*$ have the $f$ smallest values in the above sorted order).

 Denote $\calR^*=\calV-F_1^*-F_2^*$.
We have
\begin{align}
\nonumber
\sum_{i\in \calR^*} h_i^{\prime}\pth{\tx}&=\sum_{i\in [n]-F_1^*-F_2^*} h_i^{\prime}\pth{\tx}\\
&=\sum_{K=f+1}^{n-f}g_K(\tx)=0 ~~~~~~~~~~~~~~~~~~~~\text{by (\ref{alg1.3})}
\label{alternative R0}
\end{align}

%

%
%
The remaining proof is identical to the proof of Theorems \ref{t_algo1} and \ref{t_algo11} with $\bar{F}_1$ replaced by $F_1^*$, and $\bar{F}_2$ replaced by $F_2^*$.

\section{Proof of Lemma \ref{valid hull non-iter2}}\label{app: valid hull non-iter2}

\begin{proof}
Let $x_1, x_2\in \widetilde{Y}$ such that $x_1\not=x_2$. By definition of $\widetilde{Y}$, there exist valid functions $p_1(x)=\sum_{i\in \calN}\alpha_i h_i(x)\in \widetilde{\calC}$ and $p_2(x)=\sum_{i\in \calN}\beta_i h_i(x)\in \widetilde{\calC}$ such that $x_1\in \argmin~ p_1(x)$ and $x_2\in \argmin~ p_2(x)$, respectively. Note that it is possible that $p_1(\cdot)=p_2(\cdot)$, and that $p_i(\cdot)=\widetilde{p}(\cdot)$ for $i=1$ or $i=2$.\\

Given $0\le \alpha\le 1$, let $x_{\alpha}=\alpha x_1+(1-\alpha) x_2$. We consider two cases:
\begin{itemize}
\item[(i)] ~$x_{\alpha}\in \argmin ~p_1(x)\cup \argmin ~p_2(x) \cup \argmin ~ \widetilde{p}(x)$, and
\item[(ii)] $x_{\alpha}\notin \argmin ~p_1(x)\cup \argmin ~p_2(x) \cup \argmin ~ \widetilde{p}(x)$
\end{itemize}

\paragraph{{\bf Case (i)}:~$x_{\alpha}\in \argmin ~p_1(x)\cup \argmin ~p_2(x) \cup \argmin ~ \widetilde{p}(x)$}

When $x_{\alpha}\in \argmin ~p_1(x)\cup \argmin ~p_2(x) \cup \argmin ~ \widetilde{p}(x)$, by definition of $\widetilde{Y}$, we have
$$ x_{\alpha}\in \argmin ~p_1(x)\cup \argmin ~p_2(x) \cup \argmin ~ \widetilde{p}(x)\subseteq \widetilde{Y}.$$
Thus, $x_{\alpha}\in \widetilde{Y}$.

\paragraph{{\bf Case (ii)}:~$x_{\alpha}\notin \argmin ~p_1(x)\cup \argmin ~p_2(x) \cup \argmin ~ \widetilde{p}(x)$}

By symmetry, WLOG, assume that $x_1<x_2$. By definition of $x_{\alpha}$, it holds that $x_1<x_{\alpha}<x_2$. By assumption of case (ii), it must be that $x_{\alpha}> \max \pth{\argmin p_1(x)}$ and $x_{\alpha}< \min \pth{\argmin p_1(x)}$, which imply that $p_1^{\prime}(x_{\alpha})>0$ and $p_2^{\prime}(x_{\alpha})<0$.

There are two possibilities for $\widetilde{p}^{\prime}(x_{\alpha})$
(the gradient of $\widetilde{p}(x_{\alpha})$):
$\widetilde{p}^{\prime}(x_{\alpha})<0$ or
$\widetilde{p}^{\prime}(x_{\alpha})>0$.
Note that $\widetilde{p}^{\prime}(x_{\alpha})\neq 0$
because $x_\alpha\not\in\argmin ~ \widetilde{p}(x)$.

~

When $\widetilde{p}^{\prime}(x_{\alpha})<0$, there exists $0\le \zeta\le 1$ such that
$$
 \zeta ~p_1^{\prime}(x_{\alpha}) + (1-\zeta)~\widetilde{p}^{\prime}(x_{\alpha})=0.$$
By definition of $p_1(x)$ and $\widetilde{p}(x)$, we have
 \begin{align*}
 0=\zeta ~p_1^{\prime}(x_{\alpha}) + (1-\zeta)~\widetilde{p}^{\prime}(x_{\alpha})&=\zeta ~\pth{\sum_{i\in \calN}\alpha_ih_i^{\prime}(x_{\alpha})} + (1-\zeta)\pth{\frac{1}{|\calN|}\sum_{i\in \calN}h_i^{\prime}(x_{\alpha})}\\
 &=\sum_{i\in \calN}\pth{ \alpha_i\zeta +(1-\zeta)\frac{1}{|\calN|}}h_i^{\prime}(x_{\alpha}).
\end{align*}
Thus, $x_{\alpha}$ is an optimum of function
\begin{align}
\label{valid obj}
\sum_{i\in \calN}\pth{ \alpha_i\zeta +(1-\zeta)\frac{1}{|\calN|}}h_i(x).
\end{align}
Let $\calI$ be the collection of indices defined by
$$\calI\triangleq \{~i: ~ i\in \calN, ~\text{and}~ \alpha_i\zeta +(1-\zeta)\frac{1}{|\calN|}~~\ge~~ \frac{1}{2(|\calN|-f)}~\}.$$
Next we show that $|\calI|\ge |\calN|-f$. Let $\calI_1$ be the collection of indices defined by
$$\calI_1\triangleq \{~i: ~ i\in \calN, ~\text{and}~ \alpha_i\ge \frac{1}{2(|\calN|-f)}~\}.$$
Since $p_1(x)\in \widetilde{\calC}$, then $|\calI_1|\ge |\calN|-f$. In addition, since $n>3f$, $|\calN|< 2(|\calN|-f)$. 
Then, for each $j\in \calI_1$, we have
$$\alpha_i\zeta +(1-\zeta)\frac{1}{|\calN|}\ge \zeta \frac{1}{2(|\calN|-f)}+(1-\zeta)\frac{1}{|\calN|}> \zeta \frac{1}{2(|\calN|-f)}+(1-\zeta)\frac{1}{2(|\calN|-f)}= \frac{1}{2(|\calN|-f)},$$
i.e., $j\in \calI$. Thus, $\calI_1\subseteq \calI$.

Since $|\calI_1|\ge |\calN|-f$, we have $|\calI|\ge |\calN|-f$. So function (\ref{valid obj}) is a valid function. Thus, $x_{\alpha}\in \widetilde{Y}$. \\

Similarly, we can show that the above result holds when $\widetilde{p}^{\prime}(x_{\alpha})>0$.

Therefore, set $\widetilde{Y}$ is convex.

\eproof
\end{proof}

\section{Proof of Lemma \ref{gradient valid non-iter2}}\label{app: gradient valid non-iter2}

\begin{proof}
Let $\calR^*[t-1]$ denote the set of nodes from whom the remaining $n-2f$ values were received in iteration $t$, and let us denote by $\calL[t-1]$ and $\calS[t-1]$ the set of nodes from whom the largest $f$ values and the smallest $f$ values were received in iteration $t$. Due to the fact that each value is transmitted using Byzantine broadcast, $\calR^*[t-1]$, $\calL[t-1]$ and $\calS[t-1]$ do not depend on $j$, for $j\in \calN$. Thus, $\calR^*[t-1]$, $\calL[t-1]$ and $\calS[t-1]$ are well-defined.\\

Let $i^*, j^*\in \calR^*[t-1]$ such that $g_{i^*}[t-1]=\check{g}[t-1]$ and $g_{j^*}[t-1]=\hat{g}[t-1]$.
Recall that $|\calF|=\phi$. Let $\calL^*[t-1]\subseteq \calL[t-1]-\calF$ and $\calS^*[t-1]\subseteq \calS[t-1]-\calF$ such that
$$|\calL^*[t-1]|=f-\phi+|\calR^*[t-1]\cap \calF|,$$
and
$$|\calS^*[t-1]|=f-\phi+|\calR^*[t-1]\cap \calF|.$$

We consider two cases: (i) $\hat{g}[t-1]>\check{g}[t-1]$ and (ii) $\hat{g}[t-1]=\check{g}[t-1]$.

\paragraph{{\bf Case (i)}: $\hat{g}[t-1]>\check{g}[t-1]$.}

By definition of $\calL^*[t-1]$ and $\calS^*[t-1]$, we have
$$\frac{1}{f-\phi+|\calR^*[t-1]\cap \calF|}\sum_{i\in \calS^*[t-1]}g_i[t-1]\le g[t-1]\le \frac{1}{f-\phi+|\calR^*[t-1]\cap \calF|}\sum_{j\in \calL^*[t-1]}g_j[t-1].$$
Thus, there exists $0\le \xi\le 1$ such that
\begin{align}
\label{extrem nonfaulty}
\nonumber
g[t-1]&=\xi\pth{\frac{1}{f-\phi+|\calR^*[t-1]\cap \calF|}\sum_{i\in \calS^*[t-1]}g_i[t-1]}+(1-\xi)\pth{\frac{1}{f-\phi+|\calR^*[t-1]\cap \calF|}\sum_{j\in \calL^*[t-1]}g_j[t-1]}\\
&=\frac{\xi}{f-\phi+|\calR^*[t-1]\cap \calF|}\sum_{i\in \calS^*[t-1]}g_i[t-1]+\frac{1-\xi}{f-\phi+|\calR^*[t-1]\cap \calF|}\sum_{j\in \calL^*[t-1]}g_j[t-1].
\end{align}
By symmetry, WLOG, assume $\xi\ge \frac{1}{2}$. \\

Let $k\in \calR^*[t-1]-\calF$. By symmetry, WLOG, assume $g_k[t-1]\le g[t-1]$. Since $|\calL[t-1]\cup \{j^*\}|=f+1$, there exists a non-faulty agent $j^{\prime}_k\in \calL[t-1]\cup \{j^*\}$. Thus, $g_{j^{\prime}_k}[t-1]\ge \hat{g}[t-1]>g[t-1]$, and there exists $0\le \xi_k\le 1$ such that
\begin{align}
\label{middle nonfaulty}
\frac{1}{2}\pth{\hat{g}[t-1]+\check{g}[t-1]}=g[t-1]=\xi_k g_k[t-1]+(1-\xi_k) g_{j^{\prime}_k}[t-1].
\end{align}
Since $\check{g}[t-1]\le g_k[t-1]$ and $g_{j^{\prime}_k}[t-1]\ge \hat{g}[t-1]$, it can be shown that $\frac{1}{2}\le \xi_k \le 1$.
Therefore, we have
\begin{align}
\label{rewritten 1}
\nonumber
g[t-1]&=\frac{|\calN|-f}{|\calN|-f}g[t-1]\\
\nonumber
&=\frac{1}{|\calN|-f}\pth{\sum_{k\in \calR^*[t-1]-\calF}g[t-1]}+\frac{f-\phi+|\calR^*[t-1]\cap \calF|}{|\calN|-f}g[t-1]\\
\nonumber
&=\frac{1}{|\calN|-f}\sum_{k\in \calR^*[t-1]-\calF}\pth{\xi_k g_k[t-1]+(1-\xi_k) g_{j^{\prime}_k}[t-1]}\\
&\quad +\frac{\xi}{|\calN|-f}\sum_{i\in \calS^*[t-1]}g_i[t-1]+\frac{1-\xi}{|\calN|-f}\sum_{j\in \calL^*[t-1]}g_j[t-1].
\end{align}
In (\ref{rewritten 1}), for each $k\in \calR^*[t-1]-\calF$, it holds that $\frac{\xi_k}{|\calN|-f}\ge \frac{1}{2|\calN|-f}$. For each $k\in \calS^*[t-1]$, it holds that $\frac{\xi}{|\calN|-f}\ge \frac{1}{2|\calN|-f}$. In addition, we have
\begin{align*}
|\pth{\calR^*[t-1]-\calF}\cup \pth{\calS^*[t-1]}|&=|\calR^*[t-1]-\calF|+|\calS^*[t-1]|\\
&=|\calR^*[t-1]|-|\calR^*[t-1]\cap \calF|+|\calS^*[t-1]|\\
&=n-2f-|\calR^*[t-1]\cap \calF|+f-\phi+|\calR^*[t-1]\cap \calF|\\
&=n-\phi-f=|\calN|-f.
\end{align*}

Thus, in (\ref{rewritten 1}), at least $|\calN|-f$ non-faulty agents are assigned with weights lower bounded by $\frac{1}{2(|\calN|-f)}$. Thus $g[t-1]$ is a gradient of a valid function in $\calC$ at $x[t-1]$, i.e., there exists $p(x)\in \calC$ such that $g[t-1]=p^{\prime}(x[t-1]).$

\paragraph{{\bf Case (ii)}: $\hat{g}[t-1]=\check{g}[t-1]$.}
Let $k\in \calR^*[t-1]-\calF$. Since $\hat{g}[t-1]\ge g_k[t-1]\ge \check{g}[t-1]$ and $\hat{g}[t-1]=\check{g}[t-1]$, it holds that $\hat{g}[t-1]=g_k[t-1]=\check{g}[t-1]$.
Consequently, we have
$$g[t-1]=\frac{1}{2}\pth{\hat{g}[t-1]+\check{g}[t-1]}=g_k[t-1].$$
So we can rewrite $g[t-1]$ as follows.
\begin{align}
\label{rewritten 2}
\nonumber
g[t-1]&=\frac{|\calN|-f}{|\calN|-f}g[t-1]\\
\nonumber
&=\frac{1}{|\calN|-f}\pth{\sum_{k\in \calR^*[t-1]-\calF}g[t-1]}+\frac{f-\phi+|\calR^*[t-1]\cap \calF|}{|\calN|-f}g[t-1]\\
\nonumber
&=\frac{1}{|\calN|-f}\sum_{k\in \calR^*[t-1]-\calF}g_k[t-1]\\
&\quad +\frac{\xi}{|\calN|-f}\sum_{i\in \calS^*[t-1]}g_i[t-1]+\frac{1-\xi}{|\calN|-f}\sum_{j\in \calL^*[t-1]}g_j[t-1].
\end{align}
Thus, in (\ref{rewritten 2}), at least $|\calN|-f$ non-faulty agents are assigned with weights lower bounded by $\frac{1}{2(|\calN|-f)}$. Thus $g[t-1]$ is a gradient of a valid function in $\calC$ at $x[t-1]$, i.e., there exists $p(x)\in \calC$ such that $g[t-1]=p^{\prime}(x[t-1]).$

Case (i) and Case (ii) together proves the lemma.

\eproof
\end{proof}